\documentclass{article}

\usepackage{PRIMEarxiv}

\usepackage[round]{natbib}

\usepackage{microtype}
\usepackage{graphicx}

\usepackage{booktabs} 

\usepackage{hyperref}

\usepackage{algorithm}
\usepackage{algorithmic}

\usepackage{float}

\usepackage{amsmath}
\usepackage{amssymb}
\usepackage{mathtools}
\usepackage{amsthm}
\usepackage{bbm}

\usepackage{xfrac}
\usepackage{faktor}
\usepackage{centernot}
\usepackage[capitalize,noabbrev]{cleveref}

\usepackage{nicematrix}
\usepackage{makecell}

\usepackage{comment}

\usepackage{subcaption}
\usepackage{caption}

\usepackage{tikz}
\usetikzlibrary{bayesnet}
\usetikzlibrary{fit}
\usetikzlibrary{shapes,arrows}
\usetikzlibrary{shapes.multipart}
\usetikzlibrary{arrows.meta}
\usetikzlibrary{fit}

\usepackage{tikz-network}

\NiceMatrixOptions
  {
    custom-line = 
     {
       letter = : ,
       command = dashedline , 
       ccommand = cdashedline ,
       tikz = dashed
     }
  }

\DeclareMathSymbol{\shortminus}{\mathbin}{AMSa}{"39}

\theoremstyle{plain}
\newtheorem{theorem}{Theorem}[section]

\newtheorem{lemma}[theorem]{Lemma}
\newtheorem{corollary}[theorem]{Corollary}
\theoremstyle{definition}
\newtheorem{definition}[theorem]{Definition}

\theoremstyle{remark}
\newtheorem{remark}[theorem]{Remark}

\usepackage[textsize=tiny]{todonotes}

\newcommand{\ind}{\rotatebox[origin=c]{90}{$\models$}}

\newcommand{\X}{{\mathbf{X}}}

\newcommand{\cG}{{\mathcal{G}}}
\newcommand{\cH}{{\mathcal{H}}}

\newcommand{\cE}{{\mathcal{E}}}
\newcommand{\cB}{{\mathcal{B}}}

\newcommand{\cC}{{\mathcal{C}}}
\newcommand{\cM}{\mathcal{M}}

\newcommand{\cS}{{\mathcal{S}}}
\newcommand{\cT}{{\mathcal{T}}}
\newcommand{\cU}{{\mathcal{U}}}
\newcommand{\cV}{{\mathcal{V}}}

\newcommand{\doublemarked}{\circ\!\!-\!\!\circ}

\title{Separation-Based Distance Measures for Causal Graphs}

\author{
  Jonas Wahl \\
  German Research Centre for Artificial Intelligence (DFKI) \\ Saarbr\"ucken, Germany \\
  \texttt{jonas.wahl@dfki.de} \\
   \And
  Jakob Runge \\
  University of Potsdam \\ Potsdam, Germany; \\
    Center for Scalable Data Analytics and \\ Artificial Intelligence (ScaDS.AI) Dresden/Leipzig, \\ Dresden, Germany;\\
    Technische Universit\"at Berlin \\ Berlin, Germany \\
  \texttt{jakob.runge@uni-potsdam.de} \\
}

\begin{document}
\maketitle

\begin{abstract}
Assessing the accuracy of the output of causal discovery algorithms is crucial in developing and comparing novel methods. Common evaluation metrics such as the structural Hamming distance are useful for assessing individual links of causal graphs. However, many state-of-the-art causal discovery methods do not output single causal graphs, but rather their Markov equivalence classes (MECs) which encode all of the graph's separation and connection statements. In this work, we propose additional measures of distance that capture the difference in separations of two causal graphs which link-based distances are not fit to assess. The proposed distances have low polynomial time complexity and are applicable to directed acyclic graphs (DAGs) as well as to maximal ancestral graph (MAGs) that may contain bidirected edges. We complement our theoretical analysis with toy examples and empirical experiments that highlight the differences to existing comparison metrics. 

\end{abstract}

\section{INTRODUCTION}\label{sec:intro}
Inferring causal relations from observational data in the form of a \emph{causal graph} is a highly challenging task for which causality researchers have proposed numerous algorithms. While the assumptions on which these algorithms rely differ widely, many existing approaches, including the PC algorithm \citep{spirtes_causation_1993}, and its descendants stablePC \citep{colombo_order-independent_2014}, consistentPC \citep{li_constraint-based_2019}, PCMCI+ \citep{runge_discovering_2020}, LPMCI \citep{gerhardus_high-recall_2020}, FCI \citep{spirtes_causal_2013}, as well as GES \citep{chickering_optimal_2003} have in common that they are provably able to infer the correct causal graph up to \emph{Markov equivalence} in the infinite sample limit if their respective assumptions are fulfilled. More precisely, each of these methods assumes that the process that generated the data under investigation can be represented by a causal graph that belongs to a predefined class $\mathbb G$, most commonly the class of \emph{directed acyclic graphs (DAGs)} for the PC algorithm, GES, and many others.  The class of graphs $\mathbb{G}$ is equipped with a notion of \emph{separation}, for instance, $d$-separation for DAGs. Separation is a graphical criterion that specifies whether two nodes $X, Y$ in a graph are either \emph{separated} or \emph{connected} by another subset of nodes $\cS$. In practice, separation is used to translate conditional independence statements regarding the node variables of causal graphs into graphical language. Two graphs belonging to the same class $\mathbb G$ are called \emph{Markov equivalent} if they share exactly the same separation/connection statements.

Even if their inferences are correct, the methods mentioned above do not output the full `ground truth' graph underlying the data-generating process but only its Markov equivalence class (MEC), or in other words, the totality of separation/connection statements implied by the `true' graph. Some methods such as LiNGaM \citep{shimizu_linear_2006}) utilize further assumptions about functional dependencies and noise distributions to go beyond the MEC, but without such more restrictive assumptions, the MEC is all one can infer.

New causal discovery algorithms are empirically validated mainly through a range of simulations and experiments that compare their output graphs to a known ground truth. The difference between the output and the true graph is commonly measured with the structural Hamming distance (SHD) \citep{acid_searching_2003,tsamardinos_max-min_2006}, or false positive/negative edge detection rates. Such metrics can be used to assess the correctness of the presence and absence of individual adjacencies and their orientations which represent direct causal relationships. To assess the quality of an output graph with respect to the downstream task of quantifying total causal effects, further comparison metrics have been developed. The most notable example is the structural intervention distance (SID) \citep{peters_structural_2014}, which was recently reframed by \citet{henckel2024adjustment} as a special case of a so-called \emph{adjustment identification distance (AID)} alongside a sped-up algorithm for its computation and numerous other improvements and generalizations. Measures of predictive performance that have been fine-tuned to causal graphs exist as well, see e.g. \citep{liu_stability_2010,biza_tuning_2020}, but so far, causal discovery researchers do not seem to have incorporated them into their performance analyses on a broad scale.

\paragraph{Contributions} 
In this work, we propose to add another family of comparison metrics to the existing canon that compares the implied separation/connection statements of two graphs. Surprisingly, even though separation statements are  what is actually being inferred by many algorithms, they are rarely considered in evaluations, with a few notable exceptions, e.g. \citep{hyttinen2014constraint}. To fill this gap, we discuss two types of separation-based graph distances:
\begin{itemize}
    \item \textbf{s/c-metrics} formalize the simple idea to count all possible separation statements, potentially with a prescribed maximal size of the separating set, and compare their validity in both graphs under investigation. This type of comparison arguably yields the most complete picture, and produces a distance measure that is a proper metric in the mathematical sense, no matter the type of separation under discussion. However, it is of limited practical applicability due to the exponential increase of separation statements that need to be checked when the graph size grows. Even though this issue can be alleviated to a degree by working with Monte-Carlo style random approximations of the metric (see Appendix \ref{app.additional_experiments}), a heavy computational overhead remains. In summary, this type of separation-based evaluation yields the arguably most complete picture but lacks desirable scaling properties.
    \item \textbf{separation distances} provide scalable alternatives that transfer the logic of adjustment-based distances \citep{henckel2024adjustment} to separations. Instead of considering all possible separation statements in two graphs $\cG,\cH$, separation distances require the user to specify a \emph{separation strategy} that chooses a single separation set $\cS$ for every pair of separable nodes in $\cH$, and validate whether $\cS$ remains a separating set in the graph $\cG$. We present different possible separation strategies for DAGs as well as for maximal ancestral graphs (MAGs) and their MECs represented by CPDAGs and PAGs. Using the toolset developed in \citep{henckel2024adjustment}, we show that the computational complexity of separation distances is of low polynomial order in the number of nodes, making them employable even for large graphs. For MAGs, to our knowledge, these are the first examples of causal comparison metrics beyond the SHD.
\end{itemize}

We illustrate the differences between separation-based and existing graph distances with a number of toy examples and empirical simulations. We also discuss potential pitfalls in the typical link-based evaluation of causal discovery algorithms, by showing that adding or changing the orientations of few edges on a graph can lead to significant changes in the separations encoded in its MEC. Our findings highlight the importance of utilizing a broad range of evaluation metrics that take into account different causal characteristics of the output of causal discovery algorithms, from direct and total effects to the separation statements we consider here.

\section{NOTATION} \label{sec.notations}

Throughout this work, we consider causal graphs, generically denoted by $\cG$ or $\cH$, over a set $\cV$ of $N$ nodes. More specifically, we will focus on directed acyclic graphs (DAGs), maximal ancestral graphs (MAGs) and their MECs which can be represented by complete partially directed acyclic graphs (CPDAGs) for DAGs and by partial ancestral graphs (PAGs) for MAGs. We recall that a mixed graph (MG) is a graph that may contain directed $\to$ or bidirected edges $\leftrightarrow$ and that a path $\pi = (\pi(1),\pi(2),\dots,\pi(n))$ in an MG is a sequence of adjacent edges in which no non-endpoint appears twice. A non-endpoint node $C$ on a path $\pi$ in a MG is a \emph{collider} if its preceding and succeeding edge both point towards $C$, 
and a \emph{non-collider} if it is not a collider. A path $\pi$ is $m$-blocked by a subset of nodes $\cS$ if $\cS$ contains a non-collider on $\pi$ or if there is a collider on $\pi$ that does not have any descendants in $\cS$. A set $\cS$ is said to $m$-separate\footnote{When $\cG$ is a DAGs, it is more common to speak of $d$-separation rather than $m$-separation.} two nodes $(X,Y)$ on $\cG$ if it $m$-blocks all paths between them, and we denote this by  $X \bowtie_{\cG} Y | \cS$. In this case, $\cS$ will also be called a separator for $(X,Y)$. A path that is not $m$-blocked by the empty set is called $m$-open. A mixed graph contains an \emph{almost directed cycle} if there are is a directed path $X\to \dots \to Y$ and a bidirected edge $X \leftrightarrow Y$. A MAG is a mixed graph that (a) does not contain any directed or almost directed cycle (ancestral) and (b) in which any two non-adjacent nodes can be $m$-separated by some set $\cS$ (maximal). Other types of separation, such as $\sigma$-separation for cyclic graphs \citep{bongers_foundations_2021} exist as well, and many results of this work, particularly those pertaining to s/c-metrics, generalize to any class of graphs equipped with a notion of separation. The symbol $\mathbb G$ will be our generic symbol for a class of graphs on the node set $\cV$, which we hide in the notation because it stays fixed throughout this article. For example, we just write $\mathbb G = \{ \mathrm{DAGs} \}$ for the set of DAGs over $\cV$.

We also fix the following conventions:
\begin{itemize}
    \item $\cC$ denotes the set of all triples $(X,Y,\cS)$ with $X,Y \in \cV, \ X \neq Y$ and $\cS \subset \cV \backslash\{X,Y \}$. We implicitly identify triples $(X,Y,\cS)$ and $(Y,X,\cS)$. $\cC^k$ is the set of triples $(X,Y,\cS) \in \cC$ with $|\cS| = k$. In particular, $\cC = \bigsqcup_{k=0}^{N-2} \cC^k$.
    \item We write $\cC_{con}(\cG) = \{ (X,Y,\cS) \in \cC  \ | \ X \centernot{\bowtie}_{\cG} Y | \cS \}$ for the set of connection statements in a graph $\cG$ and $\cC_{sep}(\cG) = \{ (X,Y,\cS) \in \cC  \ | \ X \bowtie_{\cG} Y | \cS \}$ for the set of its separations. We also write $\cC_{con}^k(\cG)$ ($\cC_{sep}^k(\cG)$), for the set of connection (separation) statements of order $k$, i.e. with $|\cS| = k$.
    
\end{itemize}

Two causal graphs $\cG, \cH \in \mathbb G$ are called \emph{Markov equivalent} if they share the same separations, i.e. if $\cC_{sep}(\cG) = \cC_{sep}(\cH)$. This equivalence relation partitions $\mathbb G$ into disjoint Markov equivalence classes, and we write $\mathrm{MEC}(\cG) = \{ \cG' | \cG' \sim \cG \}$ for the MEC of the graph $\cG$.  
If $\mathbb G = \{ \mathrm{DAGs} \}$, any MEC can be represented uniquely as a CPDAG which may contain directed as well as undirected edges. The CPDAG representing $\mathrm{MEC}(\cG)$ is defined as the graph with the same skeleton as $\cG$ in which an edge is oriented $X \to Y$ if and only if it is oriented in this way for every class member $\cG' \in \mathrm{MEC}(\cG).$ The MECs of MAGs can be represented as PAGs which may contain directed, bidirected, undirected and semidirected ($\circ\!\!\rightarrow$) edges. The PAG representing $\mathrm{MEC}(\cG)$ is the graph with the same skeleton as $\cG$ in which an edge mark is drawn if and only if the edge mark appears in every class member $\cG' \in \mathrm{MEC}(\cG).$ If we start with a PAG or CPDAG $\cH$, we write $\mathrm{MAG}(\cH)$, respectively $\mathrm{DAG}(\cH)$ for the set of all MAGs/DAGs in the respective MEC, i.e. $\mathrm{MAG}(\cH) = \mathrm{MEC}(\cG)$ for any $\cG \in \mathrm{MAG}(\cH)$.

We also adopt the following terminology for graphical reasoning. A node $Y$ is a parent of a node $X$ in a mixed graph if there is a directed edge $Y\to X$ and a child if there is a directed edge $Y\leftarrow X$. If there is a bidirected edge $Y \leftrightarrow X$, then $Y$ is called a \emph{sibling}\footnote{These conventions unfortunately differ across the literature. We decided to employ the family tree inspired conventions in which sibling means hidden common parent and spouse means common child.} of $X$. $Y$ is an \emph{ancestor} of $X$ if there exists a directed path from $Y$ to $X$ and a \emph{descendant} of $X$ if there is a directed path from $X$ to $Y$. Two nodes that share a common child are called \emph{spouses}. We write $\mathrm{ch}_{\cG}(X),\mathrm{pa}_{\cG}(X), \mathrm{an}_{\cG}(X),\mathrm{des}_{\cG}(X), \mathrm{sib}_{\cG}(X),\mathrm{sp}_{\cG}(X)$ for the children, parents, ancestors, descendants, siblings or spouses of $X$ in $\cG$ respectively. $Y$ is a \emph{possible parent} of $X$ in a MAG $\cG$ if $Y \in \mathrm{ppa}_{\cG}(X) := \bigcup_{\cG'\in \mathrm{MEC}(\cG)} \mathrm{pa}_{\cG'}(X)$ and a \emph{possible ancestor} if $Y \in \mathrm{pan}_{\cG}(X) := \bigcup_{\cG'\in \mathrm{MEC}(\cG)} \mathrm{anc}_{\cG'}(X)$. Clearly possible parents/ancestors only depend on the MEC of a graph and can be read off its graphical representation. For instance, $Y$ is a possible parent of $X$ in a DAG $\cG$ if it is connected to $X$ by an edge $Y \to X$ or $Y-X$ in $\mathrm{CPDAG}(\cG)$. Hence we can also write  $\mathrm{ppa}_{\cG}(X)$ whenever $\cG$ is a CPDAG or a PAG. We use the shortcut notation $\mathrm{pa}_{\cG}(X\cup Y) := \mathrm{pa}_{\cG}(X)\cup \mathrm{pa}_{\cG}(Y)$, and similarly for children, spouses etc. If $\cT \subset \cV$ is a subset of nodes of a graph $\cG$, then $\cG_{\cT}$ is the subgraph of $\cG$ over $\cT$.


\section{EXISTING DISTANCE MEASURES} \label{sec.existing}

\paragraph{The Structural Hamming Distance}
The SHD \citep{acid_searching_2003,tsamardinos_max-min_2006} is the most commonly used metric to compare two causal graphs. If $\cG,\cH$ are two causal graphs, $\mathrm{SHD}(\cG,\cH)$ is defined as the number of node pairs that do not have the same type of edge between them in both graphs. For instance, the node pair $(X,Y)$ contributes to $\mathrm{SHD}(\cG,\cH)$ if $X\to Y$ in $\cG$ but there is either no edge between $X$ and $Y$ in $\cH$ or an edge that is oriented differently. The SHD is popular as it can be computed fast with only a few lines of code and it can be adapted to graphs of arbitrary types including those representing MECs. However, the SHD has also been criticized for not properly capturing the \emph{full} causal implications of the differences in the graphs $\cG$ and $\cH$, as it only assesses the presence/absence of direct effects. We illustrate this critique with the following toy example which shows that small structural changes to a graph can change its causal implications considerably.

  \begin{table*}[t!]
   \centering
   \scalebox{0.5}{
    \begin{tikzpicture}
        \Vertex[x=0,size=1.25,color=white,label=$X_1$,fontscale=1.5] {X1}  %
        \Vertex[x=3,size=1.25,color=white,label=$X_2$,fontscale=1.5] {X2} 
        \node[right=1.5cm of X2] (Dummy0) {};
          \node[right=2cm of X2] (Dummy1) {...};%
          \node[right=3cm of X2] (Dummy3) {};
          \Vertex[x=9,size=1.25,color=white,label=$X_M$,fontscale=1.5] {XM}
          \Vertex[x=12,size=1.25,color=white,label=$X_{M+1}$,fontscale=1.5] {XM1}
          \Vertex[x=15,size=1.25,color=white,label=$X_{M+2}$,fontscale=1.5] {XM2}
          \node[right=1cm of XM1] (Dummy2) {...};%
    \Vertex[x=21,size=1.25,color=white,label=$X_{2M}$,fontscale=1.5] {X2M}

     \node[right=1.5cm of XM2] (Dummy5) {};
          \node[right=2cm of XM2] (Dummy6) {...};%
          \node[right=3cm of XM2] (Dummy7) {};

     \Edge[Direct](X1)(X2)
     \Edge[Direct](X2)(Dummy0)
     \Edge[Direct](Dummy3)(XM)
     \Edge[Direct,color=red](XM)(XM1)
     \Edge[Direct](XM1)(XM2)
     \Edge[Direct](XM2)(Dummy5)
     \Edge[Direct](Dummy7)(X2M)

     \node[xshift=-1.5cm] (Dummy4) {\scalebox{1.5}{$\cG$:}};

        \Vertex[x=0,y=-2,size=1.25,color=white,label=$X_1$,fontscale=1.5] {X1}  %
        \Vertex[x=3,y=-2,size=1.25,color=white,label=$X_2$,fontscale=1.5] {X2} 
        \node[right=1.5cm of X2] (Dummy0) {};
          \node[right=2cm of X2] (Dummy1) {...};%
          \node[right=3cm of X2] (Dummy3) {};
          \Vertex[x=9,y=-2,size=1.25,color=white,label=$X_M$,fontscale=1.5] {XM}
          \Vertex[x=12,y=-2,size=1.25,color=white,label=$X_{M+1}$,fontscale=1.5] {XM1}
          \Vertex[x=15,y=-2,size=1.25,color=white,label=$X_{M+2}$,fontscale=1.5] {XM2}
          \node[right=1cm of XM1] (Dummy2) {...};%
    \Vertex[x=21,y=-2,size=1.25,color=white,label=$X_{2M}$,fontscale=1.5] {X2M}

     \node[right=1.5cm of XM2] (Dummy5) {};
          \node[right=2cm of XM2] (Dummy6) {...};%
          \node[right=3cm of XM2] (Dummy7) {};

     \Edge[Direct](X1)(X2)
     \Edge[Direct](X2)(Dummy0)
     \Edge[Direct](Dummy3)(XM)
     \Edge[Direct,color = red](XM1)(XM)
     \Edge[Direct](XM1)(XM2)
     \Edge[Direct](XM2)(Dummy5)
     \Edge[Direct](Dummy7)(X2M)

     \node[xshift=-1.5cm,yshift=-2cm] (Dummy4) {\scalebox{1.5}{$\cH$:}};

     \node[yshift=-3.5cm] (Dummy8') {};

      \end{tikzpicture}
      
      }

    \captionof{figure}{Two DAGs $\cG$ and $\cH$.} \label{fig:toy_example1}
\end{table*}


\paragraph{Example 1}
 The DAGs $\cG$ and $\cH$ in Figure \ref{fig:toy_example1} differ in the orientation of one edge only, meaning that their SHD is low relative to the total number of node pairs to the point of being negligible for a large number of nodes. On the other hand, any two nodes are connected by an open causal path in $\cG$, while in $\cH$ any node on the left of the collider $X_M$ is $d$-separated from every node on the right. In particular, no causal effects are permeated from the left to the right of the graph, and one could therefore argue that causally these graphs are not very close. In Section \ref{sec.experiments} we also show empirically that edge removal or reversal as in this example, has a much more pronounced effect on separation- and adjustment-based distance measures than on the SHD.
 


\paragraph{Adjustment Identification Distances}
AIDs \citep{henckel2024adjustment} are designed to compare DAGs with a view on causal effect estimation downstream tasks, and they include the structural intervention distance (SID) of \cite{peters_structural_2014} as a special case. The core idea of AIDs is to choose an adjustment strategy that produces identification formulas for causal effects implied by the graph $\cH$. In a second step, one then needs to verify whether the inferred identification formulas are valid in the base graph $\cG$. Whenever this is not the case, a penalty of $1$ is added to the distance score. We summarize their exact definitions in Appendix \ref{app.AIDs}.
\cite{henckel2024adjustment} provide an implementation of AIDs that achieves a computational complexity of only $\mathcal{O}(N^2)$ in sparse graphs. In addition, they propose an adaptation of AIDs to CPDAGs. However, there is no obvious relationship between the AID of two MECs represented by CPDAGs and their class members as the following example illustrates.



\paragraph{Example 2}

We first consider the two DAGs $\cG', \cH'$ in Figure \ref{fig:toy_example2} which are both causal chains, but with a different variable ordering. Clearly, these graphs are very different structurally as they do not share a single edge. This difference is witnessed by the parent-AID (= SID) when applying it to the DAGs $\cG',\cH'$ directly. However, in either of the associated CPDAGs, there are no identifiable causal effects. Therefore, the CPDAG-parent-AID as defined in \citep{henckel2024adjustment} is zero. As a consequence, an MEC inferred by a causal discovery method may score perfectly when comparing it to the CPDAG of a known ground truth DAG, while \emph{any} of the class members get a non-zero score. 
In other words, there is no guarantee that the AID between two MECs is bounded below (above) by the minimal (maximal) value obtained by their class members. \\
     Our main motivation for introducing separation distances is to define metrics that measure causal claims that are invariant under Markov equivalence. While this is desirable when the output of the discovery task are MECs, it comes at the cost that such distances can never distinguish members of the same MEC by construction. Therefore, when the output of a discovery method is more fine-grained than the MEC, separation distances should be combined with other metrics such as AIDs or the SHD. More generally, we are of the opinion that an evaluation across a broad range of metrics is desirable in most cases anyway. Nevertheless, even when more than the MEC can be inferred, a separation-based comparison is still a valid endeavor. Inferring separations is a valuable aspect of a correct discovery \emph{even if} the method aims for more.

   \begin{table*}[t!]
   \centering
   \scalebox{0.5}{
    \begin{tikzpicture}
        \Vertex[x=0,size=1.25,color=white,label=$X_1$,fontscale=1.5] {X1}  %
        \Vertex[x=2,size=1.25,color=white,label=$X_3$,fontscale=1.5] {X3} 
         \Vertex[x=4,size=1.25,color=white,label=$X_5$,fontscale=1.5] {X5}
          \node[right=1cm of X5] (Dummy1) {...};%
    \Vertex[x=8,size=1.25,color=white,label=$X_{2M\shortminus1}$,fontscale=1.5] {X2M1}

     \Vertex[x=0,y=-2,size=1.25,color=white,label=$X_2$,fontscale=1.5] {X2}  %
        \Vertex[x=2,y=-2,size=1.25,color=white,label=$X_4$,fontscale=1.5] {X4} 
         \Vertex[x=4,y=-2,size=1.25,color=white,label=$X_6$,fontscale=1.5] {X6}
          \node[right=1cm of X6] (Dummy2) {...};%
          \node[above=0.6cm of Dummy2,xshift=0.7cm] (Dummy3) {};%
          \node[above=0.8cm of Dummy2,xshift=-0.7cm] (Dummy4) {};%
    \Vertex[x=8,y=-2,size=1.25,color=white,label=$X_{{\textstyle\mathstrut}2M}$,fontscale=1.5] {X2M}
     \Edge[Direct](X1)(X2)
     \Edge[Direct](X2)(X3)
     \Edge[Direct](X3)(X4)
     \Edge[Direct](X4)(X5)
     \Edge[Direct](X5)(X6)
     \Edge[Direct](X2M1)(X2M)

     \Edge[Direct](Dummy3)(X2M1)
     \Edge[Direct](X6)(Dummy4)

     \node[xshift=-1.5cm,yshift=-1cm] (Dummy4) {\scalebox{1.5}{$\cG'$:}};

        \Vertex[x=15,size=1.25,color=white,label=$X_1$,fontscale=1.5] {X1'}  %
        \Vertex[x=17,size=1.25,color=white,label=$X_3$,fontscale=1.5] {X3'} 
         \Vertex[x=19,size=1.25,color=white,label=$X_5$,fontscale=1.5] {X5'}
          \node[right=1cm of X5'] (Dummy1') {...};%
    \Vertex[x=23,size=1.25,color=white,label=$X_{2M\shortminus1}$,fontscale=1.5] {X2M1'}

     \Vertex[x=15,y=-2,size=1.25,color=white,label=$X_2$,fontscale=1.5] {X2'}  %
        \Vertex[x=17,y=-2,size=1.25,color=white,label=$X_4$,fontscale=1.5] {X4'} 
         \Vertex[x=19,y=-2,size=1.25,color=white,label=$X_6$,fontscale=1.5] {X6'}
          \node[right=1cm of X6'] (Dummy2') {...};%
          \node[above=0.6cm of Dummy2',xshift=0.7cm] (Dummy3') {};%
          \node[above=0.8cm of Dummy2',xshift=-0.7cm] (Dummy4') {};%
    \Vertex[x=23,y=-2,size=1.25,color=white,label=$X_{{\textstyle\mathstrut}2M}$,fontscale=1.5] {X2M'}
     \Edge(X1')(X2')
     \Edge(X2')(X3')
     \Edge(X3')(X4')
     \Edge(X4')(X5')
     \Edge(X5')(X6')
     \Edge(X2M1')(X2M')

     \Edge(Dummy3')(X2M1')
     \Edge(X6')(Dummy4')

     \node[xshift=12cm,yshift=-1cm] (Dummy4') {\scalebox{1.5}{$\mathrm{CPDAG}(\cG')$:}};
     
      \end{tikzpicture}
      }
      \\
      \vspace{0.5cm}
       \centering
   \scalebox{0.5}{
    \begin{tikzpicture}
        \Vertex[x=0,size=1.25,color=white,label=$X_1$,fontscale=1.5] {X1}  %
        \Vertex[x=2,size=1.25,color=white,label=$X_3$,fontscale=1.5] {X3} 
         \Vertex[x=4,size=1.25,color=white,label=$X_5$,fontscale=1.5] {X5}
          \node[right=1cm of X5] (Dummy1) {...};%
    \Vertex[x=8,size=1.25,color=white,label=$X_{2M\shortminus1}$,fontscale=1.5] {X2M1}

     \Vertex[x=0,y=-2,size=1.25,color=white,label=$X_2$,fontscale=1.5] {X2}  %
        \Vertex[x=2,y=-2,size=1.25,color=white,label=$X_4$,fontscale=1.5] {X4} 
         \Vertex[x=4,y=-2,size=1.25,color=white,label=$X_6$,fontscale=1.5] {X6}
          \node[right=1cm of X6] (Dummy2) {...};%
          \node[left=0.1cm of Dummy1] (Dummy3) {};%
          \node[right=0.2cm of Dummy1] (Dummy4) {};%
    \Vertex[x=8,y=-2,size=1.25,color=white,label=$X_{{\textstyle\mathstrut}2M}$,fontscale=1.5] {X2M}
     \Edge[Direct](X1)(X3)
     \Edge[Direct](X3)(X5)
     \Edge[Direct](X5)(Dummy3)
     \Edge[Direct](Dummy4)(X2M1)
     \Edge[Direct](X2M1)(X2)

     \node[left=0.1cm of Dummy2] (Dummy5) {};%
          \node[right=0.2cm of Dummy2] (Dummy6) {};

     \Edge[Direct](X2)(X4)
     \Edge[Direct](X4)(X6)
     \Edge[Direct](X6)(Dummy5)
     \Edge[Direct](Dummy6)(X2M)

     \node[xshift=-1.5cm,yshift=-1cm] (Dummy7) {\scalebox{1.5}{$\cH'$:}};

     \Vertex[x=15,size=1.25,color=white,label=$X_1$,fontscale=1.5] {X1'}  %
        \Vertex[x=17,size=1.25,color=white,label=$X_3$,fontscale=1.5] {X3'} 
         \Vertex[x=19,size=1.25,color=white,label=$X_5$,fontscale=1.5] {X5'}
          \node[right=1cm of X5'] (Dummy1') {...};%
    \Vertex[x=23,size=1.25,color=white,label=$X_{2M\shortminus1}$,fontscale=1.5] {X2M1'}

     \Vertex[x=15,y=-2,size=1.25,color=white,label=$X_2$,fontscale=1.5] {X2'}  %
        \Vertex[x=17,y=-2,size=1.25,color=white,label=$X_4$,fontscale=1.5] {X4'} 
         \Vertex[x=19,y=-2,size=1.25,color=white,label=$X_6$,fontscale=1.5] {X6'}
          \node[right=1cm of X6'] (Dummy2') {...};%
          \node[left=0.1cm of Dummy1'] (Dummy3') {};%
          \node[right=0.2cm of Dummy1'] (Dummy4') {};%
    \Vertex[x=23,y=-2,size=1.25,color=white,label=$X_{{\textstyle\mathstrut}2M}$,fontscale=1.5] {X2M'}
     \Edge(X1')(X3')
     \Edge(X3')(X5')
     \Edge(X5')(Dummy3')
     \Edge(Dummy4')(X2M1')
     \Edge(X2M1')(X2')

     \node[left=0.1cm of Dummy2'] (Dummy5') {};%
          \node[right=0.2cm of Dummy2'] (Dummy6') {};

     \Edge(X2')(X4')
     \Edge(X4')(X6')
     \Edge(X6')(Dummy5')
     \Edge(Dummy6')(X2M')

     \node[xshift=12cm,yshift=-1cm] (Dummy7') {\scalebox{1.5}{$\mathrm{CPDAG}(\cH')$:}};
     \node[yshift=-3.5cm] (Dummy8') {};

      \end{tikzpicture}

      }
      

    \captionof{figure}{Two DAGs $\cG', \cH'$ and their CPDAGs.} \label{fig:toy_example2}
\end{table*}

\section{SEPARATION-BASED DISTANCES} \label{sec.distances}

\paragraph{s/c-Metric}
Arguably, the most straightforward way to compare the implied separations/connections of two causal graphs is to define a graded sum over all possible separation/connection statements. If $\cG$ is a causal graph with an appropriate notion of separation, we use the separation indicator function $\iota_{\cG}:\cC \to \{0,1 \}$,
\begin{align*}
    \iota_{\cG}(X,Y,\cS) = \begin{cases} 1 \quad &\text{if } X \centernot{\bowtie}_{\cG} Y | \cS \\  0 \quad &\text{if } X \bowtie_{\cG} Y | \cS.  \end{cases}
\end{align*}



\begin{definition} \label{def.sc-metric_new}
Consider two causal graphs $\cG,\cH$ over the same set of $N$ nodes, equipped with an appropriate notion of separation. Define
\begin{align*}
    d_{s/c}^{k}(\cG,\cH) := \frac{1}{|\cC^k|} \sum_{(X,Y,\cS) \in \cC^k} | \iota_{\cG}(X,Y,\cS) - \iota_{\cH}(X,Y,\cS)|
\end{align*}
for $k = 0,\dots, N-2$.
The \textbf{s/c-metric up to order} $\mathbf{K}$ is defined as
\begin{align*}
    d_{s/c}^{\leq K}(\cG,\cH) = \frac{1}{K+1} \sum_{k=0}^K d_{s/c}^{k}(\cG,\cH).
\end{align*}
If $K = N-2$, we speak of the \textbf{full s/c-metric} and write $d_{s/c}$ instead of $d_{s/c}^{\leq N-2}$.
\end{definition}

\begin{figure}[t!]
   \centering
    \includegraphics[width=0.6\textwidth]{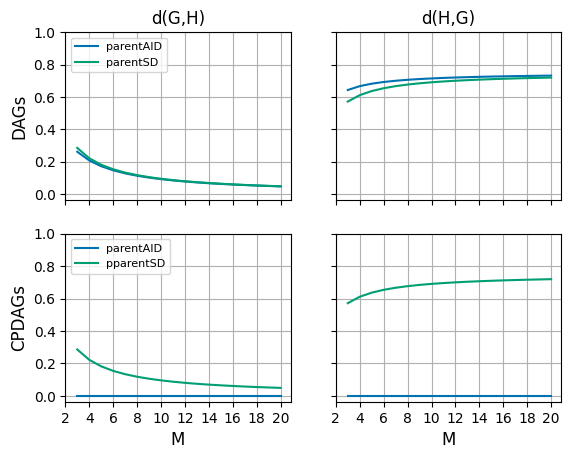}
   \caption{Empirical comparison between the parent-AID and the separation-based metrics \emph{parent-SD} and \emph{pparent-SD} to be defined in Section \ref{sec.distances} on the DAGs $\cG',\cH'$ and their CPDAGs. The label $M$ refers to the index in Figure \ref{fig:toy_example2}. The qualitative difference between the adjustment-based metric and its separation-based analog is negligible for DAGs but clearly visible for CPDAGs.\label{fig:toy_example2-evaluation}}
    \end{figure}

In words, the s/c-metric is a weighted count of the disagreement in separation/connection statements between the two graphs. It defines a mathematical metric on MECs, see Lemma \ref{lem.metric} in Appendix \ref{app.s/c_metrics}. If a graph $\cG$ is fixed as a reference point, for instance because it is the ground-truth in a simulated experiment, we can also compare the separation and connection statements implied by $\cG$ separately to those implied by $\cH$. This yields notions of \emph{false positive and false negative rate for separations}, see Appendix \ref{app.s/c_metrics} for exact definitions and a more detailed discussion. Informal versions of such error rates were already used in the simulation section of \citep{hyttinen2014constraint} but have not been adopted more broadly.
\\ \\
Separation-based measures of this kind have an obvious drawback. The number of separation statements $(X,Y,\cS)$ grows exponentially in the number of nodes so that they are slow to compute and do not scale to large graphs. Therefore, in practice, a manageable subset of separation statements needs to be selected to keep computations manageable. For this, there are two possible approaches. The first one is to choose statements randomly to compute a Monte-Carlo style approximation of the full metric. The second approach is to select only few separating statements according to a predefined deterministic strategy. We provide a plot of the quality of Monte-Carlo style approximation in Appendix \ref{app.additional_experiments}. In the main text, we focus on the second approach, since it delivers a greater computational speed-up and requires a deeper theoretical investigation. 

\paragraph{Separation Strategies}
Separation strategies are the analog of adjustment strategies \cite{henckel2024adjustment} for separation. Roughly speaking, the fundamental idea of \cite{henckel2024adjustment} is to associate to every pair of nodes $(X,Y)$ in a DAG $\cH$ a so-called adjustment set $\mathrm{ad}^{\cH}(X,Y)$ according to a fixed strategy, and to validate whether the induced adjustment formula for computing the causal effect of $X$ on $Y$ remains valid in a second DAG $\cG$. This second DAG $\cG$ is typically (but not necessarily) the ground truth in a simulated experiment. \citep{henckel2024adjustment} explicitly consider parent adjustment, ancestor adjustment, and optimal adjustment~\citep{runge2021necessary} as possible adjustment strategies. We will propose separation analogs of these strategies for DAGs, MAGs and their MECs (CPDAGs and PAGs). This task is non-trivial along two dimensions: first, the direct analogs of parent and ancestor adjustment, parent separation and ancestor separation, while valid for DAGs, are not valid separation strategies for general mixed graphs, see Figure \ref{fig.counterexample_parent_separation} in Appendix \ref{app.missing_proofs}. Based on the work \citep{van_der_zander2020finding}, we will therefore also introduce ZL-separation for MAGs, named after that paper's authors. Secondly, while being a separating set is invariant under Markov equivalence, the specific \emph{choice of separating set}, if done naively, may depend on the MEC member (more on this below).

We fix a class of graphs $\mathbb{G}$ with an appropriate notion of separation in which any pair of non-adjacent nodes can be separated. We primarily have the classes $\mathbb{G} = \{\mathrm{DAGs} \}, \{\mathrm{CPDAGs} \}, \{\mathrm{MAGs} \}, \{\mathrm{PAGs} \} $ with $m$-separation in mind, but we will briefly discuss extensions to cyclic graphs towards the end of this section.
\begin{definition} \label{def.sep_strategy}
    A \emph{separation strategy (sep-strategy)}  for $\cG\in \mathbb G$ is a map $\mathfrak{S}_{\cG}(\cdot, \cdot)$ that associates to every pair of non-adjacent nodes $(X,Y)$ of $\cG$ a set $\mathfrak{S}_{\cG}(X, Y)$ such that $X \bowtie_{\cG} Y | \mathfrak{S}_{\cG}(X, Y)$. A \emph{universal sep-strategy for} $\mathbb{G}$ is a map $\mathfrak{S}: \cG \mapsto \mathfrak{S}_{\cG} $ that associates a sep-strategy to every $\cG \in \mathbb G$.
\end{definition}

Given two graphs $\cG, \cH \in \mathbb G$ and a sep-strategy $\mathfrak{S}_{\cH}(\cdot, \cdot)$ for the latter, we define the indicator function
\begin{align*}
    \iota^{\mathfrak{S}}_{\cG,\cH}(X,Y) = \begin{cases}
         1 \qquad & \text{if }  X \centernot{\bowtie}_{\cG} Y | \mathfrak{S}_{\cH}(X, Y) \\
         0 \qquad & \text{if }  X \bowtie_{\cG} Y | \mathfrak{S}_{\cH}(X, Y). 
    \end{cases}
\end{align*}
for any pair of  $\cG$-non-adjacent nodes $(X,Y) \notin \cE_{\cH}$.

\begin{definition}[$\mathfrak{S}$-separation distance]
Let $\mathfrak{S}$ be a universal sep-strategy for $\mathbb G$.
The (normalized) $\mathfrak{S}$\emph{-separation distance} ($\mathfrak{S}$-SD) of  $\cG$ and $\cH \in \mathbb G$ is defined as
\begin{align*}
    d^{\mathfrak{S}}(\cG,\cH) := \frac{1}{N(N-1)} \sum_{(X,Y) \in \mathrm{nadj}(\cH)} \iota^{\mathfrak{S}}_{\cG,\cH}(X,Y),
\end{align*}
where the sum runs over non-adjacent node pairs in $\cH$.
\end{definition}

Thus, the $\mathfrak{S}$-SD measures whether the separating set for $(X,Y)$ on $\cH$ that has been selected according to the strategy $\mathfrak{S}$ remains a separating set on $\cG$, and if this is not the case, a penalty is incurred.

\begin{remark} \label{rem.symmetry}
 Like the adjustment identification distances of \cite{henckel2024adjustment}, separation distances are not symmetric, i.e. $ d^{\mathfrak{S}}(\cG,\cH) \neq d^{\mathfrak{S}}(\cH,\cG) $, as the two graphs play different roles in the distance's computation. We also define a symmetric $\mathfrak{S}$-SD of $\cG$ and $\cH$ by taking the geometric mean of both directions
 \begin{align*}
    d_{\text{sym}}^{\mathfrak{S}}(\cG,\cH) := \tfrac{1}{2} \left( d^{\mathfrak{S}}(\cG,\cH) + d^{\mathfrak{S}}(\cH,\cG) \right).
\end{align*}
Other symmetrizations, e.g. taking the harmonic mean instead of the geometric one, are also possible. 

\end{remark}


Since any DAG is a MAG, any universal sep-strategy for $\{\mathrm{MAGs}\}$ is also a universal sep-strategy for $\{\mathrm{DAGs}\}$. Similarly, any universal sep-strategy $\mathfrak{S}$ for $\{\mathrm{CPDAGs}\}$ defines a universal sep-strategy $\mathfrak{S}'$ for $\{\mathrm{DAGs}\}$ by $\mathfrak{S}'_{\cG} = \mathfrak{S}_{\mathrm{MEC}(\cG)}$ as any separating set for a MEC is a separating set for all its members. We note the subtlety that the converse of this statement need not be true because a sep-strategy $\mathfrak{S}$ defined on DAGs \emph{may not be well-defined} on CPDAGs, in the sense that we may have $\mathfrak{S}_{\cG}(X,Y) \neq \mathfrak{S}_{\cG'}(X,Y)$ for some nodes $X,Y$ even if $\cG \sim \cG'$ are Markov equivalent.

  



Before discussing further theoretical properties of SDs, we present several possible sep-strategies in increasing order of generality. For proofs of all results of this section, we refer to Appendix \ref{app.missing_proofs}.

\paragraph{Parent Separation on DAGs} As for adjustment, the most straightforward idea to separate non-adjacent nodes in a DAG is by using their parents. This yields the symmetric sep-strategy $\mathfrak{S}_{\cG}(X,Y) = \mathfrak{S}_{\cG}(Y,X) = \mathrm{pa}_{\cG}(X \cup Y), \ (X,Y) \notin \mathcal{E}_{\cG}$. However, parent separation is only a universal sep-strategy on DAGs. On MAGs, two non-adjacent nodes might not be separable through their parents, see Figure \ref{fig.counterexample_parent_separation} in Appendix \ref{app.missing_proofs} for a counterexample. In addition, parent separation does not directly extend to CPDAGs, as generically $\mathrm{pa}_{\cG}(X \cup Y) \neq \mathrm{pa}_{\cG'}(X \cup Y)$ for Markov equivalent $\cG \sim \cG'$. We call the corresponding separation distance for DAGs the \emph{parent-SD}.

\paragraph{Ancestor Separation on DAGs} Another universal sep-strategy for DAGs that parallels adjustment is ancestor separation, i.e. $\mathfrak{S}_{\cG}(X,Y) = \left(\mathrm{an}_{\cG}(X \cup Y) \right) \backslash \{ X,Y\}, \ (X,Y) \notin \mathcal{E}_{\cG}$. Again, this is not a valid sep-strategy for MAGs, see Figure \ref{fig.counterexample_parent_separation}, nor does it extend to CPDAGs. We call the corresponding separation distance the \emph{ancestor-SD}.

\paragraph{p-Parent Separation on CPDAGs} 
Replacing parents by possible parents yields the sep-strategy for $\mathfrak{S}_{\cG}(X,Y) = \mathrm{ppa}_{\cG}(X \cup Y), \ (X,Y) \notin \mathcal{E}_{\cG}$ whenever $\cG$ is a CPDAG. We prove that possible parents are valid separators in Appendix \ref{app.missing_proofs}, Lemma \ref{lem.possible_parents_separate}.

\paragraph{ZL-Separation on MAGs and their MECs}
To define a universal sep-strategy that remains valid on MAGs, we first recall that a set of nodes $\cS$ is a minimal separator for $(X,Y)$ on a graph $\cG$ if $X \bowtie_{\cG} Y | \cS$ and $X \centernot{\bowtie}_{\cG} Y | \cS'$ for any proper subset $\cS' \subsetneq \cS$.  \citep{van_der_zander2020finding} describe a fast algorithm that returns a minimal separator for non-adjacent nodes $X,Y$ in any MAG $\cG$. We will call this separator the \emph{ZL}-separator $\mathrm{ZL}_{\cG}(X,Y)$, and $\mathfrak{S}_{\cG}(X,Y) = \mathrm{ZL}_{\cG}(X,Y)$ defines a universal sep-strategy for MAGs. Since defining a separator as the output of an algorithm is inconvenient, we will now characterize it differently. The algorithm of \citep{van_der_zander2020finding} computes the ZL-separator by first computing a so-called \emph{nearest separator} to $(X,Y)$ and the refining this nearest separator to a minimal one. A set of nodes $\cS \subset \mathrm{pan}_{\cG}(X\cup Y)$ is called a \emph{nearest separator} relative to $(X,Y)$ on a MAG $\cG$ if (i) $X \bowtie_{\cG} Y | \cS$; and (ii) for every $W \in \mathrm{pan}_{\cG}(X\cup Y) \backslash \{X,Y\}$ and every path $\pi$ connecting $W$ and $Y$ on the moralized graph\footnote{the moralized graph of a MAG $\cH$ is the undirected graph in which nodes $A$ and $B$ share an edge if and only if on $\cH$ there is a path between them on which all nodes except $A$ and $B$ are colliders. We denote it by $\cH^m$.} $(\cG_{\mathrm{pan}_{\cG}(X \cup Y)})^m$ such that $\mathrm{nodes}(\pi) \cap \cS \neq \emptyset$, any other set $\cS' \subset \mathrm{pan}_{\cG}(X \cup Y)$ that separates $(X,Y)$ must also satisfy $\mathrm{nodes}(\pi) \cap \cS' \neq \emptyset$.


\begin{lemma} \label{lem.subset_nearest}
Let $X,Y$ be non-adjacent nodes on a MAG $\cG$, and let $\cS$ be a nearest separator for $(X,Y)$. Then any separator $\cS' \subset \cS$ is also nearest for $(X,Y)$.
\end{lemma}


\begin{theorem} \label{lem.unique_minimal_nearest}
Let $X,Y$ be non-adjacent nodes on a MAG $\cG$. The ZL-separator $\mathrm{ZL}_{\cG}(X,Y)$ is the unique separator that is both minimal and nearest for $(X,Y)$. Moreover, $\mathrm{ZL}_{\cG}(X,Y) = \mathrm{ZL}_{\cG'}(X,Y)$ for Markov equivalent MAGs $\cG \sim \cG'$.    
\end{theorem}

We call the separation distance based on ZL-separation the \emph{ZL-SD}. Due to Theorem \ref{lem.unique_minimal_nearest}, the ZL-SD is a  well-defined separation distance on MAGs and their MECs.

\paragraph{MB-enhanced Sep-Strategies}
In this paragraph, we present a way to adapt a universal sep-strategy $\mathfrak{S}$ using the \emph{Markov blanket (MB)} to produce a new universal sep-strategy $\mathrm{MB}(\mathfrak{S})$. This MB-enhancement can be carried out on any of the considered classes of graphs but its primary advantage is for DAGs/CPDAGs where the MB-enhanced strategy can be computed at lower computational cost than the original strategy, at least for sparse graphs.
To define MB-enhancement, we recall that the \emph{Markov blanket} $\mathrm{MB}_{\cG}(X)$ of a node $X$ in a MAG $\cG$ is defined as the smallest set of nodes $\cS$ with the property that $X \bowtie_{\cG} Y | \cS$ for all $Y \notin \cS$. If $\cG$ is a DAG, the Markov blanket of $X$ can be conveniently characterized as the union of its parents $\mathrm{pa}_{\cG}(X)$, its children $\mathrm{ch}_{\cG}(X)$ and its spouses $\mathrm{sp}_{\cG}(X)$, i.e. the nodes that share a common child with $X$. For general MAGs, the Markov blanket can be characterized graphically as well but this characterization is slightly more involved, see Appendix \ref{app.MB_enhanced}. We observe that $\mathrm{MB}_{\cG}(X) = \mathrm{MB}_{\cG'}(X)$ for Markov equivalent $\cG \sim \cG'$, see Lemma \ref{lem.MB_invariance}, so that the Markov blanket is well-defined on MECs.

Now, given a universal sep-strategy $\mathfrak{S}$ on a class of graphs $\mathbb G$, we define the \emph{MB-enhanced sep-strategy} as
\begin{align*}
    \mathfrak{S}_{\cG}(X,Y) =
    \begin{cases}
        \mathrm{MB}_{\cG}(X) &\text{if } Y \notin \mathrm{MB}(X) \\
        \mathfrak{S}_{\cG}(X,Y) &\text{else. } 
    \end{cases}
\end{align*}
for two non-adjacent nodes $X,Y$. Thus, we use the Markov blanket of $X$  as a separator for all nodes except for those that are non-adjacent to $X$ but part of $\mathrm{MB}_{\cG}(X)$ themselves. For these nodes, we employ the initial sep-strategy. 
The advantage of the MB-enhanced strategy over the original one is then that the MB-separator only depends on $X$ and not on $Y$ for 'most' node pairs $(X,Y)$. In DAGs/CPDAGs of bounded node degree, the number of exceptional cases is small compared to the number of nodes $N$ which allows us to leverage the results of \cite{henckel2024adjustment} to achieve a faster implementation of complexity $\mathcal{O}(N^2)$, see below. On MAGs, this speed-up is only retained under stronger assumptions, see Appendix \ref{app.MB_enhanced}.

\paragraph{Extensions to cyclic graphs}
The s/c-metric can be defined for any class of graphs under consideration as long as this class is accompanied by a fitting notion of separation. To extend separation distances to other types of graphs, it is necessary to specify a fitting sep-strategy. For mixed graphs with cycles and $\sigma$-separation, a sep-strategy can be defined by noting that two nodes in such a graph $\cG$ can be $\sigma$-separated by $\cS$ if and only if they can be $m$-separated by $\cS$ in the \emph{acyclification} $\mathrm{ac}(\cG)$ of $\cG$ \citep{bongers_foundations_2021}. Therefore, the problem of defining a sep-strategy can be reduced to defining a sep-strategy in mixed graphs. Moreover, any mixed graph can be projected onto a MAG over the same nodes with equivalent $m$-separations. Therefore, a valid sep-strategy can be defined by projecting $\mathrm{ac}(\cG)$ onto its MAG projection and employing a sep-strategy, such as ZL-separation, for MAGs.

\paragraph{Further properties of $\mathfrak{S}$-SDs}
In this section, we focus on the theoretical properties of separation distances of MAGs which are also inherited by DAGs.
We denote the undirected skeleton of a MAG $\cG$ by $\mathrm{sk}(\cG)$. We also recall that a triple of nodes $(X,Y,Z)$ is called an \emph{unshielded triple} if $Y$ is adjacent to both $X$ and $Z$ but $X$ and $Z$ are non-adjacent. The unshielded triple $(X,Y,Z)$ is an \emph{unshielded collider} if both edges have arrowheads pointing into $Y$. 
We denote the set of all unshielded colliders on a $\cG$ by $\mathcal{U}_c(\cG)$.  Finally, in a MAG $\cG$, a path $\pi$ between nodes $X$ and $Y$ is a \emph{discriminating path} for node $V$ if (i) $\pi$ includes at least three edges; (ii)$V$ is a non-endpoint of $\pi$ and adjacent to $Y$ on $\pi$; and (iii) $X$ and $Y$ are non-adjacent in $\cG$ and every node between $X$ and $V$ on $\pi$ is both a collider on $\pi$ and a parent of $Y$.
 \cite{richardson2002ancestral} showed that two MAGs $\cG,\cH$ over the same nodes are Markov equivalent if and only if (a) they share the same skeleton; (b) they share the same unshielded colliders and (c) if $\pi$ is a discriminating path for node $V$ in both graphs, then $V$ is a collider on $\pi$ in $\cG$ if and only if it is a collider on $\pi$ in $\cH$.

 We have already discussed that two Markov equivalent MAGs satisfy $d^{\mathfrak{S}}(\cG,\cH) = 0$. The next result shows to what extent graphs at zero distance are similar.

\begin{theorem} \label{thm.sep_distances_identification}
Let $\mathfrak{S}$ be a universal sep-strategy for MAGs. If $d^{\mathfrak{S}}(\cG,\cH) = 0$, then 
\begin{enumerate}
    \item[(i)] $\mathrm{sk}(\cG) \subset \mathrm{sk}(\cH)$;
    \item[(ii)] If $(X,Y,Z)$ is an adjacent triple\footnote{that is to say, $Y$ is adjacent to both $X$ and $Z$} on both graphs, and an unshielded collider on $\cH$, then it is an unshielded collider on $\cG$.
    \item[(iii)] if $\pi$ is a discriminating path for node $V$ in both graphs, and $V$ is a collider on $\pi$ in $\cH$, then it is a collider on $\pi$ in $\cG$.
\end{enumerate}
In particular, $d_{\text{sym}}^{\mathfrak{S}}(\cG,\cH) =0$ if and only if $\cG$ and $\cH$ are Markov equivalent.
\end{theorem}

\begin{figure}[h]
    \centering
    \includegraphics[height=1.6in]{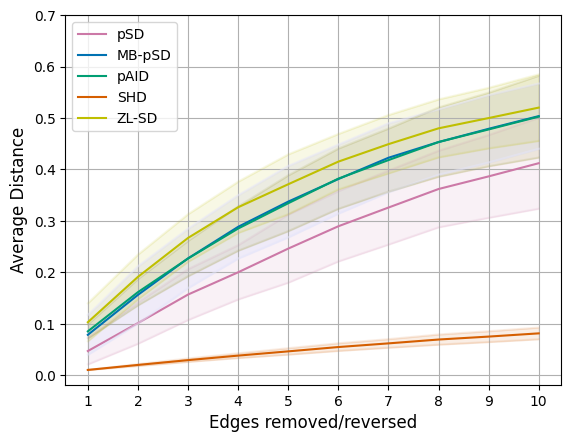}
    \caption{Effect of edge removal and reversal on different distance metrics. Parent-AID and SD are affected much more strongly by such local operations than the SHD.} \label{fig.remove_reverse}
    
\end{figure}

\begin{figure*}[t!]
    \centering
    \begin{subfigure}[t]{0.3\textwidth}
        \centering
        \includegraphics[height=1.6in]{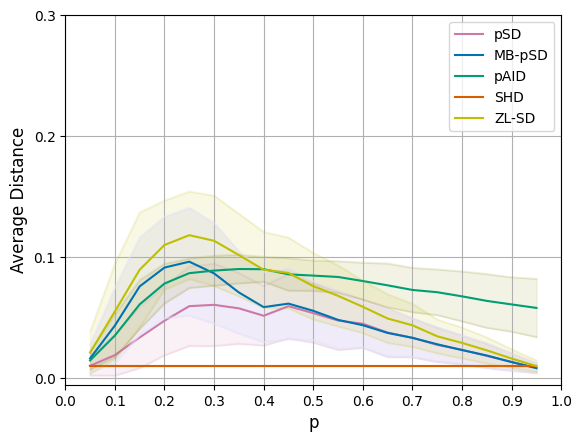}
    \end{subfigure}
    \begin{subfigure}[t]{0.3\textwidth}
        \centering
        \includegraphics[height=1.6in]{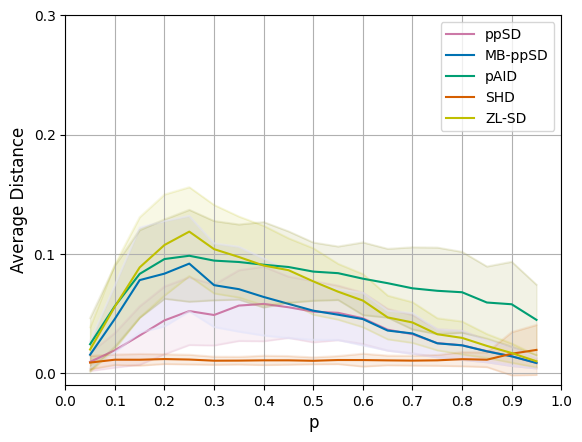}
    \end{subfigure}
    \begin{subfigure}[t]{0.3\textwidth}
        \centering
        \includegraphics[height=1.6in]{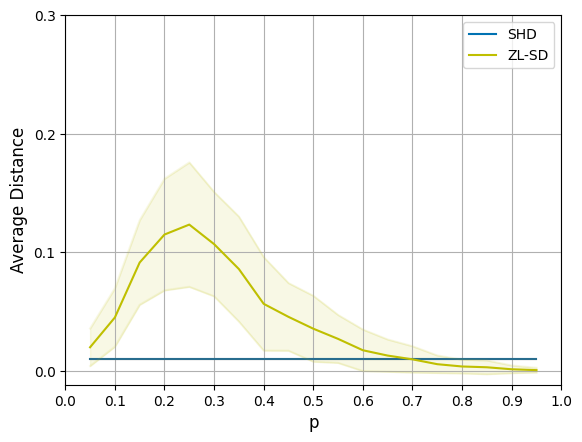}
    \end{subfigure}

    \caption{Average values of different graph distances applied to an Erdös-Renyi DAG $\cG \sim G(N=25,p)$ and another DAG $\cH$ which is obtained from $\cG$ by randomly deleting one edge and reversing one edge orientation. The plot on the left shows the distance values for the DAGs themselves while the plot in the middle compares their CPDAGs for increasing values of $p$. For mixed graphs, the only available distances are the ZL-SD and the SHD which are depicted in the rightmost plot. We ran $100$ experiments per parameter.} \label{fig.average_values}
\end{figure*}

\paragraph{Complexity of Implementations}
The computational complexity of our algorithms may depend on the sep-strategy and we provide detailed computations in Appendix \ref{app.algorithms}. In general, the distance computation consists of two steps: a separator computation step in the graph $\cH$ and a verification step on the graph $\cG$ that checks whether the separators $\mathfrak{S}_{\cH}(X,Y)$ found on $\cH$ remain separators on $\cG$. 
For non-enhanced sep-strategies excluding ZL-separation, the separator computation is $\mathcal{O}(N^3)$ or lower, for ZL-separation in a MAG this step is $\mathcal{O}(N^4)$ ($\mathcal{O}(N^3)$ on sparse graphs) \citep{van_der_zander2020finding}. The verification step can be executed with worst-case complexity $\mathcal{O}(N^2)$ ($\mathcal{O}(N)$ on sparse graphs) \citep{van2014constructing} per node pair $(X,Y)$ and thus $\mathcal{O}(N^4)$ in total (sparse: $\mathcal{O}(N^3)$). This is one degree slower than the complexity $\mathcal{O}(N^3)$ (sparse: $\mathcal{O}(N^2)$) of the algorithm that computes the parent- or ancestor-AID in \citep{henckel2024adjustment}. For MB-enhanced strategies, however, we can mimic the ideas of \cite{henckel2024adjustment} to achieve a lowered complexity of $\mathcal{O}(N^2)$ on sparse graphs. This is because, thanks to the \emph{Bayes-Ball algorithm} \citep{geiger1990d, shachter1998bayes}, see also \citep[Appendix D]{henckel2024adjustment}, we can verify many separation statements while looping over only one node $X$ instead over pairs of nodes $(X,Y)$. The loop over the second node $Y$ only has to be entered in exceptional cases. The number of exceptional cases is bounded by the size of the Markov blanket which is in turn bounded by $d^2$, where $d$ is the maximal node degree of the graph. Therefore, on sparse graphs with $d=\text{const.}$, we reach a complexity of $\mathcal{O}(N^2)$ similar to the one for AIDs. A similar trick can be applied in the separator computation step to achieve a complexity of $\mathcal{O}(N^2)$. Empirically, computing the MB-enhanced parent-SD on $100$ pairs of Erdös-Renyi graphs with $N=1000$ nodes and an expected number of $10N$ edges, took $0.59s$ on average on an Apple M1max 64GB chip as opposed to $4.35s$ for the non-enhanced parent-SD.

\section{EMPIRICAL RESULTS} \label{sec.experiments}

In Figure \ref{fig.average_values}, we draw an Erdös-Renyi DAG $G \sim G(N,p)$ from which we then randomly remove one edge and reverse the orientation of another edge to obtain a second graph $\cH$. We then compare the normalized distances of these two graphs as well as the distances of their CPDAGs. We conduct a similar experiment for mixed graphs which we generate in an Erdös-Renyi fashion with edge density $p$ and a probabilty of $0.25$ that a given edge is bidirected. Among SDs, the ZL-SD reacts the most strongly to the edge deletion/reversal. Parent-SD, MB-enhanced parent-SD and parent-AID both also witness the difference in graphs much more clearly than the SHD. On average, the parent- and pparent-SD seem to behave  more conservative than their adjustment analog and somewhat interpolate between the SHD and the parent-AID. For mixed graphs/MAGs the only available distances are the SHD and the ZL-SD which exhibit a clear difference in values for sparser graphs. Interestingly, for very dense graphs the ZL-SD even drops below the difference in SHD which is likely due to the fact that in such graphs barely any nodes can be separated anymore. In Figure \ref{fig.remove_reverse}, we plot the change in distance metrics to a base graph after we have removed an edge and reversed another one $k$-times, $k=1,\dots,10$. The base graph is a randomly drawn Erdös-Renyi DAG $G \sim G(25,0.25)$ and the removed and reversed edges are drawn uniformly random from existing edges. We have repeated this experiment $100$ times and Figure \ref{fig.remove_reverse} shows the averages across these runs. Parent-AID and all SDs are more strongly affected by these removal and reversal operations than the SHD. Remarkably, the Markov enhanced parent-SD behaves very similar to the parent-AID. In Appendix \ref{app.additional_experiments}, we provide further experiments on the correlation of different distances and on s/c-metrics.

In summary, our experiments and theoretical considerations demonstrate that AID, SDs and the SHD measure different notions of similarity of causal graphs, and the choice of metric should depend on the ultimate goal of the causal discovery effort.



\begin{table*}[t!]
\centering
\scalebox{0.8}{
\begin{NiceTabular}{|l||c:c|c|c|}
\midrule
 & \multicolumn{2}{c}{separation-based} & adjustment-based & structure-based \\
 \midrule
 & s/c-metric & SD & AID & SHD \\
 \midrule \midrule
 compares? & \makecell{all \\ separation statements} & \makecell{selected \\ separation statements} & \makecell{selected \\ adjustment sets} & \makecell{edges (presence \\ and orientation)} \\
 \midrule
 local? & no & no & no & yes \\
 \midrule
 variants & \makecell{FP-/FN-rates; \\ randomly chosen \\ sep-statements;} & \makecell{symmetrization; \\ different \\ sep-strategies;} & \makecell{symmetrization; \\ different \\ adjustment \\ strategies;} & \makecell{edge or orientation \\ FP-/FN-rates;} \\
 \midrule
 applicable to & \makecell{DAGs, MAGs, \\ CPDAGS, PAGs \\ cyclic graphs;} & \makecell{DAGs, MAGs, \\ CPDAGS, PAGs \\ cyclic graphs;} & \makecell{DAGs, CPDAGS, \\ further extensions \\ possible;} &\makecell{all graphs;}\\
 \midrule
  \makecell{scales to \\ large graphs?} & no & yes & yes & yes \\
 \midrule
 \makecell{particularly \\ recommended \\ for} & small graphs; & \makecell{MECs, \\ graph classes \\ beyond DAGs;} &  \makecell{evaluation of \\ causal discovery \\ for downstream \\ effect estimation;} & \makecell{all graphs;} \\
 \midrule
\end{NiceTabular}}
\caption{Summary of measures of comparison for causal graphs.}\label{table.comparison}
\end{table*}

\section{CONCLUSION} \label{sec.conclusion}

We have introduced new separation-based graph distances that allow to quantify how similar two causal graphs are in terms of their implied separation statements. Many of these metrics are fast to compute, scalable, and applicable to DAGs, their Markov equivalence classes (CPDAGs) as well as mixed graphs that incorporate bidirected edges for hidden confounding. We have compared them with other metrics through toy examples and empirical experiments. We summarize the properties of the available distance measures in Table \ref{table.comparison}. Our work is in line with other recent attempts to provide more comprehensive tools to evaluate causal discovery algorithms such as \citep{henckel2024adjustment,faller_self-compatibility_2023,ramsey2024choosing}. Considering these recent developments regarding evaluation together with new proposals for more realistic data simulation \citep{gamella2024causal,andrews2024better,ormaniec2024standardizing} and the use of more real-world data , we believe that it would be worthwhile to conduct an extensive re-evaluation of popular causal discovery methods across a broad range of measures and data sources in future work.

\paragraph{Code availability}

An implementation of the distance measures introduced in this work is available in the repository \href{https://github.com/JonasChoice/sep_distances}{https://github.com/JonasChoice/sep\textunderscore distances}.

\paragraph{Acknowledgements}
JW and JR received funding from the European Research Council (ERC) Starting Grant CausalEarth under the European Union’s Horizon 2020 research and innovation program (Grant Agreement No. 948112). JW also received support from the German Federal Ministry of Education and Research (BMBF) as part of the project MAC-MERLin (Grant Agreement No. 01IW24007). The authors used resources of the Deutsches Klimarechenzentrum (DKRZ) granted by its Scientific Steering Committee (WLA) under project ID bd1083.
We thank Leonard Henckel for valuable discussions and his suggestion to use Markov blankets as separators. We also thank Simon Bing for proof-reading a previous version of this manuscript.

\bibliography{Metrics}

\begin{thebibliography}{}

\bibitem[Acid and Campos, 2003]{acid_searching_2003}
Acid, S. and Campos, L. M.~d. (2003).
\newblock Searching for {Bayesian} {Network} {Structures} in the {Space} of {Restricted} {Acyclic} {Partially} {Directed} {Graphs}.
\newblock {\em Journal of Artificial Intelligence Research}, 18:445--490.

\bibitem[Andrews and Kummerfeld, 2024]{andrews2024better}
Andrews, B. and Kummerfeld, E. (2024).
\newblock Better simulations for validating causal discovery with the dag-adaptation of the onion method.
\newblock {\em arXiv preprint arXiv:2405.13100}.

\bibitem[Biza et~al., 2020]{biza_tuning_2020}
Biza, K., Tsamardinos, I., and Triantafillou, S. (2020).
\newblock Tuning {Causal} {Discovery} {Algorithms}.
\newblock In {\em Proceedings of the 10th {International} {Conference} on {Probabilistic} {Graphical} {Models}}, pages 17--28. PMLR.

\bibitem[Bongers et~al., 2021]{bongers_foundations_2021}
Bongers, S., Forré, P., Peters, J., and Mooij, J.~M. (2021).
\newblock Foundations of structural causal models with cycles and latent variables.
\newblock {\em The Annals of Statistics}, 49(5):2885--2915.

\bibitem[Chickering, 2003]{chickering_optimal_2003}
Chickering, D.~M. (2003).
\newblock Optimal structure identification with greedy search.
\newblock {\em The Journal of Machine Learning Research}, 3:507--554.

\bibitem[Claassen and Heskes, 2012]{claassen_bayesian_2012}
Claassen, T. and Heskes, T. (2012).
\newblock A {Bayesian} approach to constraint based causal inference.
\newblock In {\em Proceedings of the {Twenty}-{Eighth} {Conference} on {Uncertainty} in {Artificial} {Intelligence}}, {UAI}'12, pages 207--216, Arlington, Virginia, USA. AUAI Press.

\bibitem[Colombo and Maathuis, 2014]{colombo_order-independent_2014}
Colombo, D. and Maathuis, M.~H. (2014).
\newblock Order-{Independent} {Constraint}-{Based} {Causal} {Structure} {Learning}.
\newblock {\em Journal of Machine Learning Research}, 15(116):3921--3962.

\bibitem[Eulig et~al., 2023]{eulig_toward_2023}
Eulig, E., Mastakouri, A.~A., Blöbaum, P., Hardt, M., and Janzing, D. (2023).
\newblock Toward {Falsifying} {Causal} {Graphs} {Using} a {Permutation}-{Based} {Test}.
\newblock arXiv:2305.09565.

\bibitem[Faller and Janzing, 2025]{faller2025different}
Faller, P.~M. and Janzing, D. (2025).
\newblock On different notions of redundancy in conditional-independence-based discovery of graphical models.
\newblock {\em arXiv preprint arXiv:2502.08531}.

\bibitem[Faller et~al., 2024]{faller_self-compatibility_2023}
Faller, P.~M., Vankadara, L.~C., Mastakouri, A.~A., Locatello, F., and Janzing, D. (2024).
\newblock Self-compatibility: Evaluating causal discovery without ground truth.
\newblock In {\em International Conference on Artificial Intelligence and Statistics}, pages 4132--4140. PMLR.

\bibitem[Faltenbacher et~al., 2025]{faltenbacher2025internal}
Faltenbacher, S., Wahl, J., Herman, R., and Runge, J. (2025).
\newblock Internal incoherency scores for constraint-based causal discovery algorithms.
\newblock {\em arXiv preprint arXiv:2502.14719}.

\bibitem[Gamella et~al., 2025]{gamella2024causal}
Gamella, J.~L., Peters, J., and B{\"u}hlmann, P. (2025).
\newblock Causal chambers as a real-world physical testbed for ai methodology.
\newblock {\em Nature Machine Intelligence}, pages 1--12.

\bibitem[Geiger et~al., 1990]{geiger1990d}
Geiger, D., Verma, T., and Pearl, J. (1990).
\newblock d-separation: From theorems to algorithms.
\newblock In {\em Machine intelligence and pattern recognition}, volume~10, pages 139--148. Elsevier.

\bibitem[Gerhardus and Runge, 2020]{gerhardus_high-recall_2020}
Gerhardus, A. and Runge, J. (2020).
\newblock High-recall causal discovery for autocorrelated time series with latent confounders.
\newblock In {\em Advances in {Neural} {Information} {Processing} {Systems}}, volume~33, pages 12615--12625. Curran Associates, Inc.

\bibitem[Heisterkamp, 2009]{heisterkamp2009directed}
Heisterkamp, S.~H. (2009).
\newblock Directed acyclic graphs and the use of linear mixed models.
\newblock {\em Technical Report}.

\bibitem[Henckel et~al., 2022]{henckel2022graphical}
Henckel, L., Perkovi{\'c}, E., and Maathuis, M.~H. (2022).
\newblock Graphical criteria for efficient total effect estimation via adjustment in causal linear models.
\newblock {\em Journal of the Royal Statistical Society Series B: Statistical Methodology}, 84(2):579--599.

\bibitem[Henckel et~al., 2024]{henckel2024adjustment}
Henckel, L., W{\"u}rtzen, T., and Weichwald, S. (2024).
\newblock Adjustment identification distance: A gadjid for causal structure learning.
\newblock In {\em The 40th Conference on Uncertainty in Artificial Intelligence}.

\bibitem[Hyttinen et~al., 2014]{hyttinen2014constraint}
Hyttinen, A., Eberhardt, F., and J{\"a}rvisalo, M. (2014).
\newblock Constraint-based causal discovery: Conflict resolution with answer set programming.
\newblock In {\em UAI}, pages 340--349.

\bibitem[Kocaoglu, 2023]{kocaoglu2023characterization}
Kocaoglu, M. (2023).
\newblock Characterization and learning of causal graphs with small conditioning sets.
\newblock In {\em Thirty-seventh Conference on Neural Information Processing Systems}.

\bibitem[Li et~al., 2019]{li_constraint-based_2019}
Li, H., Cabeli, V., Sella, N., and Isambert, H. (2019).
\newblock Constraint-based causal structure learning with consistent separating sets.
\newblock In {\em Advances in Neural Information Processing Systems}, volume~32. Curran Associates, Inc.

\bibitem[Liu et~al., 2010]{liu_stability_2010}
Liu, H., Roeder, K., and Wasserman, L. (2010).
\newblock Stability approach to regularization selection (stars) for high dimensional graphical models.
\newblock In {\em Advances in Neural Information Processing Systems}, volume~23. Curran Associates, Inc.

\bibitem[Machlanski et~al., 2024]{machlanski_robustness_2023}
Machlanski, D., Samothrakis, S., and Clarke, P.~S. (2024).
\newblock Robustness of algorithms for causal structure learning to hyperparameter choice.
\newblock In {\em Causal Learning and Reasoning}, pages 703--739. PMLR.

\bibitem[Mian et~al., 2021]{mian_discovering_2021}
Mian, O.~A., Marx, A., and Vreeken, J. (2021).
\newblock Discovering {Fully} {Oriented} {Causal} {Networks}.
\newblock {\em Proceedings of the AAAI Conference on Artificial Intelligence}, 35(10):8975--8982.

\bibitem[Ormaniec et~al., 2024]{ormaniec2024standardizing}
Ormaniec, W., Sussex, S., Lorch, L., Sch{\"o}lkopf, B., and Krause, A. (2024).
\newblock Standardizing structural causal models.
\newblock {\em arXiv preprint arXiv:2406.11601}.

\bibitem[Pellet and Elisseeff, 2008]{pellet2008finding}
Pellet, J.-P. and Elisseeff, A. (2008).
\newblock Finding latent causes in causal networks: an efficient approach based on markov blankets.
\newblock {\em Advances in Neural Information Processing Systems}, 21.

\bibitem[Peters and Bühlmann, 2015]{peters_structural_2014}
Peters, J. and Bühlmann, P. (2015).
\newblock Structural intervention distance for evaluating causal graphs.
\newblock {\em Neural computation}, 27(3):771--799.
\newblock Publisher: MIT Press.

\bibitem[Pitchforth and Mengersen, 2013]{pitchforth_proposed_2013}
Pitchforth, J. and Mengersen, K. (2013).
\newblock A proposed validation framework for expert elicited {Bayesian} {Networks}.
\newblock {\em Expert Systems with Applications}, 40(1):162--167.

\bibitem[Ramsey et~al., 2006]{ramsey_adjacency-faithfulness_2006}
Ramsey, J., Spirtes, P., and Zhang, J. (2006).
\newblock Adjacency-faithfulness and conservative causal inference.
\newblock In {\em Proceedings of the {Twenty}-{Second} {Conference} on {Uncertainty} in {Artificial} {Intelligence}}, {UAI}'06, pages 401--408. AUAI Press.

\bibitem[Ramsey et~al., 2024]{ramsey2024choosing}
Ramsey, J.~D., Andrews, B., and Spirtes, P. (2024).
\newblock Choosing dag models using markov and minimal edge count in the absence of ground truth.
\newblock {\em arXiv preprint arXiv:2409.20187}.

\bibitem[Richardson and Spirtes, 2002]{richardson2002ancestral}
Richardson, T. and Spirtes, P. (2002).
\newblock {Ancestral graph Markov models}.
\newblock {\em The Annals of Statistics}, 30(4):962 -- 1030.

\bibitem[Runge, 2020]{runge_discovering_2020}
Runge, J. (2020).
\newblock Discovering contemporaneous and lagged causal relations in autocorrelated nonlinear time series datasets.
\newblock In {\em Proceedings of the 36th {Conference} on {Uncertainty} in {Artificial} {Intelligence} ({UAI})}, pages 1388--1397. PMLR.

\bibitem[Runge, 2021]{runge2021necessary}
Runge, J. (2021).
\newblock Necessary and sufficient graphical conditions for optimal adjustment sets in causal graphical models with hidden variables.
\newblock In Ranzato, M., Beygelzimer, A., Dauphin, Y., Liang, P., and Vaughan, J.~W., editors, {\em Advances in Neural Information Processing Systems 34 (NeurIPS 2021)}.

\bibitem[Shachter, 1998]{shachter1998bayes}
Shachter, R.~D. (1998).
\newblock Bayes-ball: Rational pastime (for determining irrelevance and requisite information in belief networks and influence diagrams).
\newblock In {\em Proceedings of the Fourteenth conference on Uncertainty in artificial intelligence}, pages 480--487.

\bibitem[Shimizu et~al., 2006]{shimizu_linear_2006}
Shimizu, S., Hoyer, P.~O., Hyvärinen, A., Kerminen, A., and Jordan, M. (2006).
\newblock A linear non-{Gaussian} acyclic model for causal discovery.
\newblock {\em Journal of Machine Learning Research}, 7(10).

\bibitem[Spirtes et~al., 1993]{spirtes_causation_1993}
Spirtes, P., Glymour, C., and Scheines, R. (1993).
\newblock {\em Causation, {Prediction}, and {Search}}, volume~81 of {\em Lecture {Notes} in {Statistics}}.
\newblock Springer, New York, NY.

\bibitem[Spirtes et~al., 1995]{spirtes_causal_2013}
Spirtes, P., Meek, C., and Richardson, T. (1995).
\newblock Causal inference in the presence of latent variables and selection bias.
\newblock In {\em Proceedings of the Eleventh conference on Uncertainty in artificial intelligence}, pages 499--506.

\bibitem[Textor et~al., 2016]{textor_robust_2016}
Textor, J., van~der Zander, B., Gilthorpe, M.~S., Liskiewicz, M., and Ellison, G.~T. (2016).
\newblock Robust causal inference using directed acyclic graphs: the {R} package 'dagitty'.
\newblock {\em International Journal of Epidemiology}, 45(6):1887--1894.

\bibitem[Tsamardinos et~al., 2006]{tsamardinos_max-min_2006}
Tsamardinos, I., Brown, L.~E., and Aliferis, C.~F. (2006).
\newblock The max-min hill-climbing {Bayesian} network structure learning algorithm.
\newblock {\em Machine Learning}, 65(1):31--78.

\bibitem[van~der Zander and Li\'{s}kiewicz, 2020]{van_der_zander2020finding}
van~der Zander, B. and Li\'{s}kiewicz, M. (2020).
\newblock Finding minimal d-separators in linear time and applications.
\newblock In Adams, R.~P. and Gogate, V., editors, {\em Proceedings of The 35th Uncertainty in Artificial Intelligence Conference}, volume 115 of {\em Proceedings of Machine Learning Research}, pages 637--647. PMLR.

\bibitem[van~der Zander et~al., 2014]{van2014constructing}
van~der Zander, B., Liskiewicz, M., and Textor, J. (2014).
\newblock Constructing separators and adjustment sets in ancestral graphs.
\newblock In {\em UAI 2014 Workshop Causal Inference: Learning and Prediction}, page~11.

\bibitem[Wien{\"o}bst, 2023]{wienoebst2023moralization}
Wien{\"o}bst, M. (2023).
\newblock On the computational complexity of graph moralization.
\newblock {\em Blog post at mwien.github.io}.

\bibitem[Zhang, 2008]{zhang_causal_2008}
Zhang, J. (2008).
\newblock Causal reasoning with ancestral graphs.
\newblock {\em Journal of Machine Learning Research}, 9:1437--1474.

\bibitem[Zheng et~al., 2018]{zheng_dags_2018}
Zheng, X., Aragam, B., Ravikumar, P.~K., and Xing, E.~P. (2018).
\newblock {DAGs} with {NO} {TEARS}: {Continuous} {Optimization} for {Structure} {Learning}.
\newblock In {\em Advances in {Neural} {Information} {Processing} {Systems}}, volume~31. Curran Associates, Inc.

\end{thebibliography}
\newpage
\appendix
\thispagestyle{empty}
\begin{center}
\hrule height 3pt
\vspace{0.5cm}
\Large{\textbf{Supplementary Materials}}
\vspace{0.5cm}
\hrule height 1pt
\end{center}
\vspace{0.8cm}


\section{RELATED WORK} \label{sec.related_work}
The articles \citep{acid_searching_2003,tsamardinos_max-min_2006} introduce the SHD, \citep{peters_structural_2014} introduce the SID. The latter is generalized to a broader class of adjustment identification distances in \citep{henckel2024adjustment}. \cite{liu_stability_2010} and \cite{biza_tuning_2020} define algorithms that judge causal graphs based on their predictive abilities for the purpose of hyperparameter selection. In addition, several works address method evaluation and causal graph falsification. \citet{faller_self-compatibility_2023} propose a criterion for evaluating causal discovery algorithms in the absence of a ground truth which quantifies the compatibility of causal graphs that were inferred over different subsets of variables. \citet{pitchforth_proposed_2013} review methods to validate expert-elicited Bayesian networks. \citet{machlanski_robustness_2023} discuss causal model evaluation in the context of causal effect estimation, and \citet{eulig_toward_2023} develop a permutation-based test for causal graph falsification. \cite{ramsey2024choosing} propose a statistical test for the validity of the Markov property of a distribution on a causal graph. \cite{faltenbacher2025internal} investigate how coherently the separations of the output graph of a constraint-based causal discovery method reflect the conditional independencies that where measured during its run. \cite{faller2025different} investigate how to use redundant test results to correct errors in the learned graph. 

\section{PROOFS} \label{app.missing_proofs}

\begin{lemma} \label{lem.possible_parents_separate}
Let $\cG$ be a DAG. For any pair of non-adjacent nodes $(X,Y)$, the set of possible parents $\mathrm{ppa}_{\cG}(X \cup Y)$ is a $d$-separating set.     
\end{lemma}

To prove Lemma \ref{lem.possible_parents_separate} we recall that separations on CPDAGs can be nicely handled with \emph{definite status paths}. Fix a CPDAG $\cH$ and a path $\pi = (\pi(1),\dots, \pi(n))$ on $\cH$. A non-collider $\pi(i)$ is called a \emph{definite non-collider} on $\pi$ if the triple $(\pi(i-1),\pi(i),\pi(i+1))$ is unshielded or if one of its adjacent edges is oriented away from $\pi(i)$ in $\cH$. An arbitrary node on $\pi(j)$ on $\pi$ is said to be of \emph{definite status on} $\pi$ if it is a definite non-collider or a collider. Finally the path $\pi$ is called a \emph{definite status path} if all of its non-endpoint nodes are of definite status. A definite status path $\pi$ from $X$ to $Y$ is called $d$-blocked by a set of nodes $\cS \subset \cV\backslash \{X,Y\}$ if one of its non-colliders is in $\cS$ or if one of its colliders does not have any descendants (including itself) in $\cS$. \citet{zhang_causal_2008} proved that $X$ and $Y$ are $d$-separated by a set $\cS$ in one (every) DAG in $\mathrm{DAG}(\cH)$ if and only if every definite status path on $\cH$ is $d$-blocked by $\cS$.

\begin{proof}[Proof of Lemma \ref{lem.possible_parents_separate}]
 We consider $\cH := \mathrm{CPDAG}(\cG)$, consider the pair of non-adjacent nodes $(X,Y)$ and the set of possible parents  $\cS := \mathrm{ppa}_{\cH}(X \cup Y)$. By the previous considerations, we need to show that $\cS$ $d$-blocks any definite status path between $X$ and $Y$. Let $\pi = (\pi(1),\pi(2),\dots,\pi(n)), n \geq 3$ be such a definite status path with $\pi(1) =X$ and $\pi(n) = Y$. We will need to consider several different cases. First, if $\pi(1) \leftarrow \pi(2)$ or $\pi(1) \doublemarked \pi(2)$ and similarly if $\pi(n-1) \rightarrow \pi(n)$ or $\pi(n-1) \doublemarked \pi(n)$, then $\cS$ contains a non-collider of $\pi$ and hence $\pi$ is $d$-blocked. \\
 Therefore, we now assume that $\pi(1) \rightarrow \pi(2)$ and $\pi(n-1) \leftarrow \pi(n)$. Then, there must be a collider on $\pi$. Next, we assume that $\cS$ $d$-unblocks $\pi$ and derive a contradiction.
\\ \\
\textbf{Case 1}: We first consider the case, where there is exactly one collider $C = \pi(i), 1<i<n$ on $\pi$. Then every edge $(\pi(k-1),\pi(k))$ must be directed towards $C$, i.e. $\pi(k-1) \rightarrow \pi(k)$ for $k\leq i$ and $\pi(k) \leftarrow \pi(k+1)$ for $k\geq i$ as any other orientation would generate an additional collider on $\pi$ in any DAG $\cG' \in \mathrm{DAG}(\cH)$. Since $\cS$ $d$-unblocks $\pi$, there must be a possible parent $P \in \mathrm{ppa}_{\cH}(X) \cap \mathrm{des}(C)$. Then the path $X = \pi(1) \to \pi(2) \to \dots C \to \dots P \doublemarked X$ would induce a cycle in some DAG $\cG' \in \mathrm{DAG}(\cH)$. The case $P \in \mathrm{ppa}_{\cH}(Y) \cap \mathrm{des}(C)$ is analogues.
\\ \\
\textbf{Case 2}: Now we assume that there is more than one collider on $\pi$. Let $C = \pi(i_{C})$ be the first and $C' = \pi(i_{C'})$ be the last collider on the path , so that in particular $i_C < i_{C'}$. By the same reasoning as in Case 1, the subpath $(\pi(1),\dots, \pi(i_C)$ must be right-directed and the subpath $(\pi(i_{C'}),\dots, \pi(n))$ must be left-directed. If $\cS$ $d$-unblocks $\pi$, then there must exist $P \in \cS \cap \left(\mathrm{des}(C) \cup \{C\} \right)$ and $P' \in \cS \cap \left(\mathrm{des}(C') \cup \{C'\} \right)$ with descending paths $\xi$ and $\xi'$ from $C$ to $P$ and $C'$ to $P'$ respectively. If $P \in \mathrm{ppa}_{\cH}(X)$, then the concatenation  of $(\pi(1),\dots, \pi(i_C))$, $\xi$ and the edge $P \doublemarked X$ (or $P \to X$) would yield a cycle in some $\cG' \in \mathrm{DAG}(\cH)$. By the same argument, $P' \in \mathrm{ppa}_{\cH}(Y)$ would induce a cycle in some $\cG' \in \mathrm{DAG}(\cH)$. Therefore, we must have $P \in \mathrm{ppa}_{\cH}(Y)$ and $P' \in \mathrm{ppa}_{\cH}(X)$. We now assume that $P$ is the first node on the descending path $\xi$ that is in $\mathrm{ppa}_{\cH}(Y)$ (otherwise, we replace $P$ by the first node  $P^*$ on $\xi$ with this property). Similarly we assume that $P'$ is the first node on the descending path $\xi'$ that is in $\mathrm{ppa}_{\cH}(X)$. Consider the path $\Pi$ obtained by concatenating in this order the path $(\pi(1),\dots, \pi(i_C))$, the path $\xi = (\xi(1),\dots,\xi(k))$, the edge $(P,Y)$, the inverse of the path $(\pi(i_{C'}),\dots, \pi(n))$, the path $\xi' = (\xi'(1),\dots,\xi'(l))$ and the edge $(P',X)$. We will show the contradictory conclusion that $\Pi$ must be a cycle.

If the edge $(P,Y)$ was left-directed $P \leftarrow Y$ in \emph{some} $\cG' \in \mathrm{DAG}(\cH)$, then $\xi'(l-1)\to P = \xi(l) \leftarrow Y$ would be a collider and we show that this is contradictory. Indeed, if this collider was unshielded, then $P$ would be a child of $Y$ in every $\cG' \in \mathrm{DAG}(\cH)$ which contradicts the assumption $P \in \mathrm{ppa}_{\cH}(Y)$. However, if the collider was shielded, there would have to be an edge $(\xi(l-1),Y)$ and since $P$ was the first node on $\xi$ that is a possible parent of $Y$, this edge would have to be directed as $\xi(l-1) \leftarrow Y$. But then the edge $(P,Y)$ must be directed as $P \leftarrow Y$ in \emph{every} $\cG'' \in \mathrm{DAG}(\cH)$ to avoid the cycle $Y \to \xi(l-1) \to P \to Y$. But this again contradicts the assumption $P \in \mathrm{ppa}_{\cH}(Y)$. Consequently, the edge $(P,Y)$ must be right-directed $P \to Y$ in \emph{every} $\cG' \in \mathrm{DAG}(\cH)$. We can repeat the same line of reasoning for the edge $(P',X)$ to see that this edge needs to be right-directed as well. But then we have achieved our final contradiction that the path $\Pi$ must be circular.
\\ \\
In both cases we have derived the desired contradiction and hence $\cS$ must $d$-block $\pi$.
 
\end{proof}

\begin{lemma} \label{lem.subset_nearestv2}
Let $X,Y$ be non-adjacent nodes on a MAG $\cG$, and let $\cS$ be a nearest separator for $(X,Y)$. Then any separator $\cS' \subset \cS$ is also nearest for $(X,Y)$.
\end{lemma}

\begin{proof}
It suffices to show that $\cS'$ satisfies condition (ii) in the definition of nearest separator. Let $W \in \mathrm{pan}_{\cG}(X\cup Y) \backslash \{X,Y\}$ and let $\pi$ be a path connecting $W$ and $Y$ on the moralized graph $(\cG_{\mathrm{pan}_{\cG}(X \cup Y)})^m$ with $\mathrm{nodes}(\pi) \cap \cS' \neq \emptyset$. Since $\cS' \subset \cS$, we also have $\mathrm{nodes}(\pi) \cap \cS \neq \emptyset$, and since $\cS$ is nearest, we must have $\mathrm{nodes}(\pi) \cap \cT \neq \emptyset$ for any separating set $\cT \subset \mathrm{pan}_{\cG}(X \cup Y)$, establishing property (ii) for $\cS'$.   
\end{proof}

\begin{theorem} \label{lem.unique_minimal_nearestv2}
Let $X,Y$ be non-adjacent nodes on a MAG $\cG$. The ZL-separator $\mathrm{ZL}_{\cG}(X,Y)$ is the unique separator that is both minimal and nearest for $(X,Y)$. Moreover, $\mathrm{ZL}_{\cG}(X,Y) = \mathrm{ZL}_{\cG'}(X,Y)$ for Markov equivalent MAGs $\cG \sim \cG'$.    
\end{theorem}

\begin{proof}
We first show that $\mathrm{ZL}_{\cG}(X,Y)$ is both minimal and nearest for $(X,Y)$. Minimality was already established in \citep{van_der_zander2020finding}, and by construction of the algorithm that outputs $\mathrm{ZL}_{\cG}(X,Y)$ in \citep{van_der_zander2020finding}, $\mathrm{ZL}_{\cG}(X,Y)$ is a subset of a nearest separator. By Lemma \ref{lem.subset_nearestv2}, $\mathrm{ZL}_{\cG}(X,Y)$ is also nearest.
To prove uniqueness, let $\cS,\cS'$ be two minimal nearest separators for $(X,Y)$. By symmetry we need only show $\cS \subset \cS'$. Let $W \in \cS$. We want to show that $W \in \cS'$. Since $\cS$ is a minimal separator, there must exist a path $\Pi$ from $X$ to $Y$ on $\cG$ that is not closed by $\cS\backslash \{W \}$ but closed by $\cS$. In particular, $W$ is the unique element of $\cS$ that is a non-collider on $\Pi$. By \citep[Lemma 3.17]{richardson2002ancestral} the sequence of non-colliders on $\Pi$ forms a undirected path $\pi$ on $(\cG_{\mathrm{pan}_{\cG}(X \cup Y)})^m$ from $X$ to $Y$. The subpath $\pi'$ starting from $W$ and ending at $Y$ has the property that $\mathrm{nodes}(\pi') \cap \cS = \{ W\}$.
Now, we use the assumption that $\cS$ is a nearest separator for $(X,Y)$. Condition (ii) of the definition of nearest separators namely implies that $\mathrm{nodes}(\pi') \cap \cS' \neq \emptyset $. If $W \in \mathrm{nodes}(\pi') \cap \cS'$, then in particular $W \in \cS'$ and we are done. To finish the proof, we show that $W \notin \mathrm{nodes}(\pi') \cap \cS'$ leads to a contradiction. If $W \notin \mathrm{nodes}(\pi') \cap \cS'$, we choose $W' \in \mathrm{nodes}(\pi') \cap \cS'$ and consider the subpath $\pi''$ starting from $W'$. Since $\cS'$ was also assumed to be nearest, $\mathrm{nodes}(\pi'') \cap \cS' \neq \emptyset$ implies that also $\mathrm{nodes}(\pi'') \cap \cS \neq \emptyset$. But $\pi''$ was a subpath of $\pi' \backslash \{W \}$ so it follows that  $\mathrm{nodes}(\pi') \backslash\{ W\} \cap \cS \neq \emptyset$ which contradicts $\mathrm{nodes}(\pi') \cap \cS = \{ W\}$. This concludes the proof.

Finally, we note that minimality and being nearest are invariant under Markov equivalent which proves that $\mathrm{ZL}_{\cG}(X,Y) = \mathrm{ZL}_{\cG'}(X,Y)$ for Markov equivalent MAGs $\cG \sim \cG'$.  
\end{proof}    

\begin{theorem} \label{thm.sep_distances_identification_v2}
Let $\mathfrak{S}$ be a universal sep-strategy for MAGs. If $d^{\mathfrak{S}}(\cG,\cH) = 0$, then 
\begin{enumerate}
    \item $\mathrm{sk}(\cG) \subset \mathrm{sk}(\cH)$;
    \item If $(X,Y,Z)$ is an adjacent triple on both graphs, and an unshielded collider on $\cH$, then it is an unshielded collider on $\cG$.
    \item if $\pi$ is a discriminating path for node $V$ in both graphs, and $V$ is a collider on $\pi$ in $\cH$, then it is a collider on $\pi$ in $\cG$.
\end{enumerate}
In particular, $d_{\text{sym}}^{\mathfrak{S}}(\cG,\cH) =0$ if and only if $\cG$ and $\cH$ are Markov equivalent.
\end{theorem}

\begin{proof}
    If $X$ and $Y$ are non-adjacent in $\cH$, they are separated by $\mathfrak{S}_{\cH}(X,Y)$, and since $d^{\mathfrak{S}}(\cG,\cH) = 0$, $\mathfrak{S}_{\cH}(X,Y)$ also separates them in $\cG$. Hence $X$ and $Y$ are also non-adjacent in $\cG$, proving the first claim. To establish the second claim, consider an unshielded collider $(X,Y,Z) \in \mathcal{U}_c(\cH)$ that is also an adjacent triple on $\cG$. By the first claim  $(X,Y,Z)$ is also unshielded in $\cG$. $(X,Y,Z) \in \mathcal{U}_c(\cH)$ implies that $Y \notin \mathfrak{S}_{\cH}(X,Z)$. Since $d^{\mathfrak{S}}(\cG,\cH) = 0$, $\mathfrak{S}_{\cH}(X,Z)$ is therefore a separating set not containing $Y$ on $\cG$ as well, so that $(X,Y,Z)$ must be a collider on $\cG$. Now, let $\pi = (X,W_1,\dots,W_s,V,Y)$ be a discriminating path for node $V$ between nodes $X$ and $Y$ on both graphs such that $V$ is a collider on $\pi$ in $\cH$. By induction over $i$, $W_i\in \mathfrak{S}_{\cH}(X,Y)$ for $i=1,\dots,s$ and moreover because of the collider property $V \notin \mathfrak{S}_{\cH}(X,Y)$. If $V$ was not a collider on $\pi$ in $\cG$, then $\mathfrak{S}_{\cH}(X,Y)$ would open the path $\pi$ in $\cG$ which would contradict $d^{\mathfrak{S}}(\cG,\cH) = 0$. Hence $V$ must be a collider on $\pi$ in $\cG$.

    Finally, let us consider the symmetrized distance $d_{\text{sym}}^{\mathfrak{S}}(\cG,\cH)$ which is zero if and only if $d^{\mathfrak{S}}(\cG,\cH) = 0$ and $d^{\mathfrak{S}}(\cH,\cG) =0$. So, by the first part of the theorem, both graphs have the same skeleton, the same unshielded triples and the same colliders on shared discriminating paths. Hence they are Markov equivalent. Conversely, if $\cG$ and $\cH$ are Markov equivalent, then any separating strategy for one graph is a separating strategy for the other, so $d_{\text{sym}}^{\mathfrak{S}}(\cG,\cH)=0$.
\end{proof}

The following figure proves that neither parent nor ancestor separation are valid sep-strategies for MAGs.

\begin{table}[h]
    \centering
     \tikz{ %
        \node[latent] (V) {$V$} ; %
        \node[latent,left=of V] (X) {$X$} ; %
        \node[latent,right=of V] (Y) {$Y$} ;
        \node[latent,yshift = 2cm] (Z) {$Z$} ; %
        \node[latent,yshift = -2cm] (W) {$W$} ;%
        \edge {X,V} {Z} ; %
        \edge {Z} {X,Y,V} ; %
        \edge {V,Y} {W} ;
        \edge {W} {Y,X} ;
      }

    \captionof{figure}{In this MAG $\cG$, $\mathrm{pa}_{\cG}(X \cup Y) = \{ W,Z \}$ does not $m$-separate $X$ and $Y$ as it unblocks the path $X \leftrightarrow Z \leftrightarrow V \to W \leftrightarrow Y$. The same is true for the ancestral set $\mathrm{anc}_{\cG}(X \cup Y)\backslash\{X,Y\}$ as well as for the set of potential parents $\mathrm{ppa}_{\cG}(X \cup Y)$  as both coincide with $\mathrm{pa}_{\cG}(X \cup Y)$ in this example.} \label{fig.counterexample_parent_separation}
\end{table}

\section{MB-SEPARATION IN MAGS} \label{app.MB_enhanced}

In this section, we extend the Markov blanket separation distance to MAGs. Recall that the Markov blanket of $X$ is the smallest set of nodes of $\cG$ with the property that $X \bowtie_{\cG} Y | \mathrm{MB}_{\cG}(X)$ for all $Y \notin \mathrm{MB}_{\cG}(X)$. The following characterization of the Markov blanket for MAGs has been considered in various places, see for instance \citep{pellet2008finding}. It is based on the notion of collider paths: A path $\pi = (\pi(1),\pi(2),\dots,\pi(n))$ from $X = \pi(1)$ to $Y = \pi(n)$ is a \emph{collider path} if $n > 2$ and all middle nodes $\pi(2),\dots, \pi(n-1)$ are colliders on $\pi$. The Markov blanket $\mathrm{MB}_{\cG}(X)$ of a node $X$ on a MAG $\cG$ then consists of all nodes $Y$ such that
\begin{itemize}
    \item[(i)] $Y$ is adjacent to $X$ or,
    \item [(ii)] there exists a collider path from $Y$ to $X$.
\end{itemize}

We record that the Markov blanket is invariant under Markov equivalence.

\begin{lemma} \label{lem.MB_invariance}
If two MAGs $\cG \sim \cG'$ are Markov equivalent, then for any node $X$, we have $\mathrm{MB}_{\cG}(X) = \mathrm{MB}_{\cG'}(X)$.
\end{lemma}

\begin{proof}
Since $m$-separation is invariant under Markov equivalence $\mathrm{MB}_{\cG}(X)$ has the property that $X \bowtie_{\cG'} Y | \mathrm{MB}_{\cG}(X)$ for all $Y \notin \mathrm{MB}_{\cG}(X)$. Since $\mathrm{MB}_{\cG'}(X)$ is the smallest subset with this property, we must have $\mathrm{MB}_{\cG'}(X) \subset \mathrm{MB}_{\cG}(X)$. Repeating the argument with reversed roles of $\cG$ and $\cG'$ finishes the proof.
\end{proof}

Therefore if $\mathfrak{S}$ defines a sep-strategy  on Markov equivalence classes, then so does its MB-enhanced strategy. For PAGs (and consequently also for MAGs), we can therefore define the MB-enhancement of ZL-separation
\begin{align*}
    \mathfrak{S}_{\cG}(X,Y) =
    \begin{cases}
        \mathrm{MB}_{\cG}(X) &\text{if } Y \notin \mathrm{MB}_{\cG}(X) \\
         \mathrm{ZL}_{\cG}(X,Y) &\text{else. } 
    \end{cases}
\end{align*}
for all non-adjacent pairs $(X,Y)$. In DAGs, the computational advantage of Markov blanket separation over other sep-strategies is that its implementation can avoid a double loop over nodes $X$ and $Y$ for all but the exceptional cases, leading to lower computional complexity if the number of exceptions is small, see Appendix \ref{app.algorithms} below. This is the case if the maximal node degree $d$, i.e. the maximal number of neighbors per node is small compared to the total number of nodes: any $Y \in \mathrm{MB}_{\cG}(X)$ that is non-adjacent to $X$ can be reached in exactly two steps, so the number of exceptional cases is $d^2$ at most. In MAGs, this guarantee can no longer be upheld, since collider paths can be of arbitrary lengths and thus a small node degree is no longer sufficient to bound the size of the set of exceptions. Such bounds can only be retained if in addition to the node degree, the maximal lengths of collider paths is assumed to be small compared to the number of nodes.

\section{FURTHER RESULTS ON S/C-METRICS} \label{app.s/c_metrics}

Two causal graphs $\cG,\cH$ of the same type are called $K$\emph{-th order Markov equivalent} if their separations/connections coincide up to order $K$ \citep{kocaoglu2023characterization}, that is if $d_{s/c}^{\leq K}(\cG,\cH) =0$. 

\begin{lemma} \label{lem.K_classes_metric}
Let $\mathbb G$ be a class of graphs with an appropriate notion of separation, e.g. $\mathbb G = \{ \text{MAGs}\}$ and $m$-separation. $d_{s/c}^{\leq K}$ is a metric on the set $\faktor{\mathbb G}{\sim}$ of $K$-th order Markov equivalence classes of $\mathbb G$.   
\end{lemma}

\begin{corollary} \label{lem.metric}
Consider a class of graphs $\mathbb{G}$ with an appropriate notion of separation. 

Then, $d_{s/c}(\cG,\cH) = 0$ if and only if $\cG$ and $\cH$ are Markov equivalent. Moreover, $d_{s/c}$ defines a metric (in the mathematical sense of the word) on the set of Markov equivalence classes over $\mathbb G$. 
\end{corollary}

\begin{proof}[Proof of Lemma \ref{lem.K_classes_metric}]
 Consider the real vector space $V_k$ of maps $\cC_k \to \mathbb R$ on which $\| f \|_k = \frac{1}{\cC_k}\sum_{(X,Y,\cS) \in \cC_k} |f_k(X,Y,\cS)| $ defines a norm. Then $\| (f_0,\dots,f_K)\| := \sum_{k=0}^K \| f_k \|_k$ defines a norm on the graded vector space $V_K= \bigoplus_{k=0}^{K} V_k$.
 We now define the mapping
 \begin{align*}
     g_K: \mathbb G \to V_K, \qquad \cG \mapsto (1-i_{\cG}|_{\cC_0},\dots, 1-i_{\cG}|_{\cC_K}),
 \end{align*}
 where $i_{\cG}|_{\cC_k}$ is the separation indication function restricted to the set of $k$-th order separation/connection statements. We observe that $g_K(\cG_0) = g_K(\cG)$ if and only if $\cG$ and $\cH$ have the same separations up to order $K$. In other words, the mapping $g_K$ is well-defined on the quotient $\faktor{\mathbb G}{\sim}$ of $K$-th order Markov equivalence classes and becomes an embedding.
 Since $d_{s/c}^{\leq K}(\cG,\cH)$ is nothing but 
\begin{align*}
 d_{s/c}^{\leq K}(\cG,\cH) = \| g_K(\cG)- g_K(\cH)\|,
\end{align*}
it follows directly that $d_{s/c}^{\leq K}$ is a metric on $\faktor{\mathbb G}{\sim}$. Finally, since the usual notion of Markov equivalence means that two graphs share exactly the same separations, $\cG$ and $\cH$ are Markov equivalent if and only if $d_{s/c}(\cG,\cH) = d_{s/c}^{\leq N-2}(\cG,\cH)  = 0$, so this is indeed the special case where $K = N-2$ by which point all separation statements have been exhausted.

\end{proof}

\begin{remark}

\begin{itemize}
  \item[(i)] Since $ d_{s/c}^{k}(\cG,\cH)$ is bounded by $1$, $ d_{s/c}^{\leq K}(\cG,\cH)$ is also bounded by $1$ due to the normalization constant $\tfrac{1}{K+1}$. This value is taken if $\cG$ is the fully disconnected and $\cH$ is a fully connected graph.
    \item[(ii)] If one prefers to assign more importance to differences in low order statements than to differences in higher order statements, this can be reflected by introducing a weight $w_k, \ k=0,\dots N-2$ and replace the s/c-metric with a weighted version
    \begin{align*}
    d_{s/c}^{w}(\cG,\cH) = \frac{1}{\sum_k w_k} \sum_{k=0}^K w_k\cdot d_{s/c}^{k}(\cG,\cH).
\end{align*}
\end{itemize}
\end{remark}

\subsection{Markov and Faithfulness metric} \label{subsec.Markov_Faithfulness_distance}
The s/c-metric introduced in the previous subsection is a symmetric notion of distance for two causal graphs to which differences in connections and differences in separations between the two graphs contribute equally. We can also consider separations and connections separately at the price of losing symmetry: we have to specify one graph as a reference point for the separations or connections that we would like to compare. A version of the distance measures that we are about to introduced (without the grading by order) has been used implicitly in \citep{hyttinen2014constraint} to evaluate their causal discovery method.

\begin{definition} \label{def.c-dis}
Consider two causal graphs $\cG,\cH$ over the same set of nodes, equipped with an appropriate notion of separation. 
We first define
\begin{align*}
    d_{c}^{k}(\cG,\cH) := \frac{1}{|\cC^k_{con}(\cG)|} \sum_{(X,Y,\cS) \in \cC^k_{con}(\cG)} (1- \iota_{\cH}(X,Y,\cS)).
\end{align*}
for $k = 0,\dots, N-2$. Then we call 
\begin{align*}
d_c^{\leq K}  = \frac{1}{K+1} \sum_{k=0}^K d_{c}^{k}(\cG,\cH).    
\end{align*}
 the \textbf{c-metric} or \textbf{Markov metric} of $\cH$ to $\cG$ \textbf{of order} $\mathbf{K}$. If $K = N-2$, we just speak of the \textbf{c-metric} or \textbf{Markov metric} and write $d_{c}$ instead of $d_{c}^{N-2}$.
\end{definition}

We recall that a distribution $P_{\cV}$ over the node variables $\cV$ is called \emph{Markovian} on a causal graph $\cG$, if the graphical separation $X \bowtie_{\cG} Y | \cS$ implies the conditional independence $X \ind_{P_{\cV}} Y | \cS$, and is called \emph{faithful} on $\cG$ if the graphical connection $X \centernot{\bowtie_{\cG}} Y | \cS$ implies the conditional dependence $X \centernot{\ind}_{P_{\cV}} Y | \cS$ 

The name 'Markov metric' is justified by the following result. 

\begin{lemma} \label{lem.Markov}
Consider two causal graphs $\cG,\cH$ over the same set of nodes $\X$, equipped with an appropriate notion of separation. Suppose that $P_{\cV}$ is a distribution on $\X$ that is Markovian and faithful on the graph $\cG$.
Then $d_{c}(\cG,\cH) \neq 0$, if and only if $P_{\cV}$ is not Markovian on $\cH$.
\end{lemma}

\begin{proof}
If $d_{c}(\cG,\cH) \neq 0$, then there must be a triple $(X,Y,\cS) \in \cC_{con}(\cG)\cap \cC_{sep}(\cH)$. Since  $(X,Y,\cS) \in \cC_{con}(\cG)$ and $P_{\cV}$ is faithful on $\cG$, we must have that $X \centernot{\ind}_{P_{\cV}} Y | \cS$. If $P_{\cV}$ would be Markovian on $\cH$, $(X,Y,\cS) \in \cC_{sep}(\cH)$ would imply $X\ind_{P_{\cV}} Y | \cS$, which is a contradiction. Conversely, if $P_{\cV}$ is not Markovian on $\cH$, there must be a triple $(X,Y,\cS) \in \cC_{sep}(\cH)$ with $X \centernot{\ind}_{P_{\cV}} Y | \cS$. Since $P_{\cV}$ is Markovian on $\cG$, it follows that $(X,Y,\cS) \in \cC_{con}(\cG)$ and thus $(X,Y,\cS) \in \cC_{con}(\cG)\cap \cC_{sep}(\cH)$. Hence $d_{c}(\cG,\cH) \neq 0$. 
\end{proof}

Consider a simulated experiment to evaluate a causal discovery method $\cM$ which outputs the graph $\cH$. If the parameters of the data-generation process are chosen in such a way that the Markov property and causal Faithfulness w.r.t to the ground truth graph $\cG$ is guaranteed, according to the previous lemma, $d_{c}(\cG,\cH) \neq 0$ means that $P_{\cV}$ does not have the Markov property on the output graph. In other words, $d_{c}(\cG,\cH)$ measures how far the pair $(\cH,P_{\cV})$ is from being Markovian. As the Markov property is the most fundamental link between the distribution and the graph, a large Markov metric should be a clear warning sign that the algorithm is not performing well or that the data-generation process is far from being faithful on $\cG$. Markovianity on $\cG$ is usually a given in a data simulation.

We can define an analogous notion for separations instead of connections.

\begin{definition} \label{def.s-dis}
Consider two causal graphs $\cG,\cH$ over the same set of nodes $\X$, equipped with an appropriate notion of separation.
We first define
\begin{align*}
    d_{s}^{k}(\cG,\cH) := \frac{1}{|\cC^k_{sep}(\cG)|} \sum_{(X,Y,\cS) \in \cC^k_{sep}(\cG)} \iota_{\cH}(X,Y,\cS).
\end{align*}
for $k = 0,\dots, N-2$. Then we call 
\begin{align*}
d_s^{\leq K} = \frac{1}{K+1} \sum_{k=0}^K d_{s}^{k}(\cG,\cH).    
\end{align*}
 the \textbf{s-metric} or \textbf{Faithfulness metric} of $\cH$ to $\cG$ \textbf{of order} $\mathbf{K}$. If $K = N-2$, we just speak of the \textbf{c-} or \textbf{Faithfulness metric} and write $d_{c}$ instead of $d_{c}^{N-2}$.
\end{definition}

The following result is the analogue of Lemma \ref{lem.Markov} for the Faithfulness metric.

\begin{lemma} \label{lem.faithful}
Consider two causal graphs $\cG,\cH$ over the same set of nodes $\X$, equipped with an appropriate notion of separation. Suppose that $P_{\cV}$ is a distribution on $\X$ that is Markovian and faithful on the graph $\cG$.
Then $d_{s}(\cG,\cH) \neq 0$ if and only if $P_{\cV}$ is not faithful on $\cH$.
\end{lemma}

\begin{proof}
If $d_{s}(\cG,\cH) \neq 0$, then there must be a triple $(X,Y,\cS) \in \cC_{sep}(\cG)\cap \cC_{con}(\cH)$. Since  $(X,Y,\cS) \in \cC_{sep}(\cG)$ and $P_{\cV}$ is Markovian on $\cG$, we must have that $X \ind_{P_{\cV}} Y | \cS$. If $P_{\cV}$ would be faithful on $\cH$, $(X,Y,\cS) \in \cC_{con}(\cH)$ would imply $X \centernot{\ind}_{P_{\cV}} Y | \cS$, which is a contradiction. Conversely, if $P_{\cV}$ is not faithful on $\cH$, there must a triple $(X,Y,\cS) \in \cC_{con}(\cH)$ with $X \ind_{P_{\cV}} Y | \cS$. Since $P_{\cV}$ is faithful on $\cG$, it follows that $(X,Y,\cS) \in \cC_{sep}(\cG)$ and thus $(X,Y,\cS) \in \cC_{sep}(\cG)\cap \cC_{con}(\cH)$. But this means that $d_{s}(\cG,\cH) \neq 0$.
\end{proof}

The name 'Faithfulness metric' is thus motivated by the fact that if $\cG$ is a ground truth graph in a simulated experiment, the Faithfulness metric measures the amount of Faithfulness violations in a method's output graph.

\begin{remark} \label{rem.distances}
\ 

\vspace{-0.4cm}
\begin{itemize}
    \item[(i)] Like the s/c-metric before, both the Markov and the Faithfulness metric only depend on the Markov equivalence classes of the two graphs in play.
    \item[(ii)] At first glance, it might seem like $d_{s}(\cG,\cH) = d_{c}(\cH,\cG)$ but this is not the case. The crucial difference lies in the normalization constants of the $k$-th order contributions which need not coincide in $d_{s}(\cG,\cH)$ and $d_{c}(\cH,\cG)$. 
\end{itemize}    
\end{remark}


\vspace{-0.4cm}
\paragraph{Analogy to False Positive and False Negative Rate}
The term $d_c^k$ in the Markov metric measures the number of connections of the ground truth graph $\cG$ with $|\cS| = k$ that the graph $\cH$ misses (\emph{"false negatives"}) relative to the total number of connections of $\cG$ (\emph{"positives"}) with $|\cS| = k$. Thus, since we scale the Markov metric by $\frac{1}{N-1}$ we compute the average \emph{false negative rate} across all orders for separation/connection statements. Similarly, the Faithfulness metric measures the number of separations of the ground truth graph $\cG$ with $|\cS| = k$ that the graph $\cH$ mistakes for connections (\emph{"false positives"}), relative to the total number of separations of $\cG$ (\emph{"negatives"}) with $|\cS| = k$. The Faithfulness metric, can thus be interpreted as the average \emph{false positive rate} across all orders for separation/connection statements. Like usual FPRs and FNRs, we can also combine Markov and Faithful distance into a ROC-curve to obtain an additional quality metric, see \citep{hyttinen2014constraint}.

\section{A SUMMARY OF ADJUSTMENT IDENTIFICATION DISTANCES} \label{app.AIDs}

Adjustment Identification Distances (AIDs) were introduced in \citep{henckel2024adjustment}. As we often refer to these distances in the main document, we repeat their definition here for the convenience of the reader. This section does not contain any original results, and the reader is encouraged to consult \citep{henckel2024adjustment} for more details.

AIDs are based on the notion of identifying formulas in causal graphs. An identifying formula for the effect of variable $X$ on variable $Y$ in a DAG $\cG$ is an equation that expresses the interventional distribution $P(Y | \mathrm{do}(X))$ purely in terms of the observational distribution of the graphical nodes for any distribution $P$ compatible with $\cG$. An effect is called identifiable if there is at least one identifying formula for it. A (sound and complete) identification strategy is then defined as an algorithm that inputs a tuple $(\cG,X,Y)$ and that returns a correct identifying formula if there is one, and \textsc{none} otherwise.

More specifically \cite{henckel2024adjustment} focus on identification through adjustment. If $X,Y \in \cV$ are nodes and $\cS \subset \cV \backslash \{X,Y \}$ is a subset of nodes, then $\cS$ is a valid adjustment set for the effect of $X$ on $Y$ if  $P(Y | \mathrm{do}(X)) = \int P(y|x,\mathbf{s}) P(\mathbf{s}) d \mathbf{s}$ for any distribution $P$ compatible with $\cG$, now assuming that $P$ has a density with respect to the Lebesgue measure. Since $\mathrm{pa}_{\cG}(X)$ is a valid adjustment set for $(X,Y)$ whenever $Y$ is not itself a parent of $X$, the first option to define an identification strategy is to use the identification formula obtained through parent adjustment; this is what \cite{henckel2024adjustment} call the parent adjustment strategy. Other identification strategies used in \citep{henckel2024adjustment} are ancestor adjustment, which employs the identifying formula obtained by conditioning on all ancestors of $X$, and optimal adjustment \citep{henckel2022graphical} which makes use of the adjustment set of minimal variance. To build a distance measure for DAGs from identification strategies, the final missing ingredient is a verifier, i.e. an algorithm that inputs a tuple $(\cG,X,Y)$ and an identifying formula and that outputs whether this formula is, in fact, a correct identifying formula for the effect of $X$ on $Y$ on $\cG$. 

With these tools at hand, the adjustment identification distance $\mathrm{AID}(\cG,\cH)$ for a given adjustment strategy is defined by (1) computing the identifying formula prescribed by the chosen strategy in $\cH$ for each pair of nodes $(X,Y), \ X \neq Y$; (2) verifying for each pair of nodes whether this formula is correct in $\cG$; and (3) incurring a penalty of $1$ whenever the formula is false in $\cG$.

To generalize AIDs to CPDAGs, it is first necessary to observe that in CPDAGs causal effects may longer be identifiable. Therefore, in CPDAG-AIDs, a penalty is incurred not only when the identifying formula computed in $\cH$ is incorrect in $\cG$, but also if an effect is identifiable in one graph but not in the other.

\section{ALGORITHMS} \label{app.algorithms}

\subsection{Algorithms to compute separation distances} 

In this section, we will present pseudocode to compute the parent-, the MB- and the ZL-separation distance, including a more detailed discussion of their computational complexity. Each algorithm consists of two steps, a separator computation step in one graph and a separator verification in the other. The computational bottleneck is the verification step which will therefore be our main focus.
\paragraph{Separator Computation}

Separators can be computed as follows. We record the algorithmic complexity of each step, and our usage of 'sparse' refers to a bounded node degree $d$ independent of the number of nodes $N$. Recall that a graph is of bounded node degree $d$ if each node is adjacent to at most $d$ other nodes. The number of edges of a graph will be denoted by $M$.

Computing the parent separators $\mathrm{pa}_{\cH}(X\cup Y)$ for all node pairs $(X,Y)$ is of computational complexity $\mathcal{O}(N^3)$ in general: one needs to compute $\mathrm{pa}_{\cH}(X)$ for all $X$ in a first loop ($\mathcal{O}(N^2)$) and then take the union $\mathrm{pa}_{\cH}(X\cup Y)$ ($\mathcal{O}(N)$) for all node pairs $(X,Y)$ ($\mathcal{O}(N^2)$ times), so $\mathcal{O}(N^3)$ in total. In the sparse case, taking the union is only $\mathcal{O}(1)$ so that the complexity reduces to $\mathcal{O}(N^2)$. Similar arguments also yield the same general complexity $\mathcal{O}(N^3)$ for computing ancestor and potential parent separators. Computing the $ZL$-separator of $(X,Y)$ in a MAG $\cH$ is of complexity $\mathcal{O}(N+M)$ \citep{van_der_zander2020finding} and this needs to be executed $\mathcal{O}(N^2)$ times, so that the complexity in terms of $N$ is $\mathcal{O}(N^4)$ in general and $\mathcal{O}(N^3)$ in the sparse case. Computing the Markov blanket $MB_{\cH}(X)$ in a DAG $\cH$ for all $X$ can be done by computing the moralized graph ($\mathcal{O}(N^3)$\footnote{this can actually be reduced to $\approx\mathcal{O}(N^{2.37})$, see \citep{heisterkamp2009directed,wienoebst2023moralization}} and $\mathcal{O}(N^2)$ in the sparse case) and by subsequently computing the adjacencies of $X$ in this graph for all $X$ $\mathcal{O}(N^2)$, yielding an upper bound of $\mathcal{O}(N^3)$ in general and $\mathcal{O}(N^2)$ in the sparse case. Computing the MB-enhanced (possible) parent separator is therefore of complexity $\mathcal{O}(N^3)$ in general and $\mathcal{O}(N^2)$ in the sparse case as well.

To compute the ZL-separator, we use the algorithm introduced in \citep{van_der_zander2020finding} which is of computational complexity $\mathcal{O}(N+M)$. The ZL-separator $\mathrm{ZL}_{\cG}(X,Y)$ depends on both arguments $X,Y$ and hence we need to loop through both, leading to a complexity of $\mathcal{O}(N^4)$ or $\mathcal{O}(N^3)$ if the input graph is sparse.

\paragraph{Separation Verifier}

Since verifying whether the separators computed in a graph $\cH$ are in fact separators in another graph $\cG$ is the computationally most costly step, we will provide more details here. In general, the complexity of this part of the distance computation is $\mathcal{O}(N^2\cdot(N+M))$ or $\mathcal{O}(N^4)$ in terms of $N$ only, as Algorithm \ref{alg.sep_verifier} illustrates. In the sparse case where $M \in \mathcal{O}(N)$, this becomes $\mathcal{O}(N^2)$. We write $\mathrm{nadj}(\cH)$ for the set of non-adjacent node pairs in a graph $\cH$ (computable in $\mathcal{O}(N^2)$). 

\begin{algorithm}[ht!] 
\begin{algorithmic}
\REQUIRE causal graph $\cG$, list $\mathfrak{S}_{\cH}(X,Y), (X,Y) \in  \mathrm{nadj}(\cH)$ of proposed separator found in a previous step based on another graph $\cH$.
\STATE Initialize distance $d \gets 0$;
\FOR{$(X,Y) \in \mathrm{nadj}(\cH)$ \hfill \COMMENT{loop of length $\mathcal{O}(N^2)$} \\}
\STATE check $X \bowtie_{\cG} Y | \mathfrak{S}_{\cH}(X,Y) $; \hfill \COMMENT{complexity $\mathcal{O}(N+M)$}
\IF{check returns \emph{false}}
\STATE $d \mathrel{+}= 1$;
\ENDIF
\ENDFOR
\end{algorithmic}
\caption{Pseudocode to verify separators and compute the corresponding non-normalized SD.} \label{alg.sep_verifier}.
\end{algorithm}

On a DAG or CPDAG $\cG$, MB-enhancement yields an improvement thanks to the \emph{Bayes-Ball algorithm} \citep{geiger1990d,shachter1998bayes}, see also \citep[Appendix D]{henckel2024adjustment}. For a given node $X$ and a set $\cS$, the Bayes-Ball algorithm is able to compute the set $\mathrm{nsep}_{\cG}(X,\cS)= \{ Y \ | \ Y \centernot{\bowtie}_{\cG} X | \cS  \}$ in time $\mathcal{O}(N+M)$. We apply it in Algorithm \ref{alg.sep_verifier_MB} to achieve the desired reduction in complexity.

\begin{algorithm}[ht!] 
\begin{algorithmic}
\REQUIRE DAG $\cG$, list  of Markov blankets $\mathrm{MB}_{\cH}(X), X\in  \cV$, list of separators for exceptional cases $\mathfrak{S}_{\cH}(X,Y), Y \in \mathrm{MB}_{\cH}(X)\backslash \left( \mathrm{pa}_{\cH}(X) \cup \mathrm{ch}_{\cH}(X) \right)$.\\
\STATE Initialize distance $d \gets 0$;
\ \\
\FOR{$X \in \cV$ \hfill \COMMENT{loop of length $N$}\\ }
\STATE get $\mathrm{nsep}_{\cG}(X,\cS)$ with Bayes-Ball; \hfill \COMMENT{complexity $\mathcal{O}(N+M)$ and $\mathcal{O}(N)$ if $\cG$ is sparse}
\STATE $d \gets |\mathrm{nsep}_{\cG}(X,\cS) \cap \cV \backslash\mathrm{MB}_{\cH}(X)|$; \hfill \COMMENT{complexity $\mathcal{O}(N)$}
\FOR{$Y \in \mathrm{MB}_{\cH}(X)\backslash \left( \mathrm{pa}_{\cH}(X) \cup \mathrm{ch}_{\cH}(X) \right)$; \hfill \COMMENT{loop of length $\mathcal{O}(N), \mathcal{O}(1)$ if $\cH$ is sparse} \\ }
\STATE check $X \bowtie_{\cG} Y | \mathfrak{S}_{\cH}(X,Y)$; \hfill \COMMENT{complexity $\mathcal{O}(N+M)$, $\mathcal{O}(N)$ if $\cG$ is sparse}
\IF{check returns \emph{false}}
\STATE $d \mathrel{+}= 1$;
\ENDIF
\ENDFOR
\ENDFOR
\end{algorithmic}
\caption{Pseudocode to verify MB-enhanced separators and to compute the corresponding non-normalized SD. The worst case computational complexity is $\mathcal{O}(N^2 \cdot(N+M))$. If $\cG$ and $\cH$ are sparse, the computational complexity reduces to $\mathcal{O}(N^2)$. } \label{alg.sep_verifier_MB}.
\end{algorithm}



\subsection{Algorithms to compute s/c-metrics} \label{app.algorithms_s/c_metrics}

For the sake of completion, even though the computation is straightforward, we also provide pseudocode to compute the s/c-metric. Once again, we rely on an \emph{oracle function} $\iota_{\cG}(X,Y,\cS)$ that takes in a causal graph and a triple $(X,Y,\cS) \in \cC$ and outputs $0$ if $X$ and $Y$ are separated and $1$ if they are connected by $\cS$ in $\cG$. Practical implementations of such an oracle for DAGs ($d$-separation) are available in the R package \emph{Dagitty} \citep{textor_robust_2016} and the Python packages \emph{networkx} and \emph{Tigramite} (\url{https://github.com/jakobrunge/tigramite}). The oracle implemented in Tigramite is also applicable to MAGs ($m$-separation), tsMAGs and tsDAGs.

\begin{algorithm}[ht!] 
\begin{algorithmic}
\REQUIRE $\cG, \ \cH$ causal graphs (e.g. DAGs) or Markov equivalence classes (e.g. CPDAGs) with $N$ nodes, $\mathrm{Oracle}$ for fitting notion of separation. Sets $\mathcal{L}_k$ of order $k$-triples to be tested for $k=0,\dots,N-2.$
\IF{$\cG$ (or $\cH$) is Markov equivalence class}
\STATE $\cG$ (or $\cH$) $\gets$ member of $\cG$ (or $\cH$); 
\ENDIF
\STATE $\mathrm{dist} = 0.0$;
\FOR{$k =0,\dots, N-2$}
\STATE $\mathrm{dist}_k = 0.0$;
\STATE $\mathrm{count} = 0$;
\FOR{triples $(X,Y,\cS) \in \mathcal{L}_k$}
\STATE $\mathrm{count} \mathrel{+}= 1$;
\STATE $\mathrm{dist}_k \mathrel{+}=|\iota_{\cG}(X,Y,\cS) - \iota_{\cH}(X,Y,\cS)| $;
\ENDFOR
\STATE $\mathrm{dist} \mathrel{+}= \frac{\mathrm{dist}_k}{\mathrm{count}}$;
\ENDFOR
\STATE \textbf{return} $\frac{\mathrm{dist}}{N-1}$.
\end{algorithmic}
\caption{Pseudocode to compute the s/c-metric.} \label{alg.sc_distance}.
\end{algorithm}







\section{THE CHALLENGE OF INVALID OUTPUT GRAPHS} \label{sec.challenges}
In this section, we would like to draw attention to an additional challenge that occurs when applying evaluation metrics to causal discovery algorithms based on the PC-algorithm. 
Separation-based (or adjustment-based) metrics that compare two graphs $\cG$ and $\cH$ require that both of these graphs are able to decide whether $(X,Y,\cS)$ is a separation or a connection statement (or whether $\cS$ is adjustment set for $(X,Y)$) in the respective graph. For instance, if the applied notion of separation is $d$-separation, then $\cG$ and $\cH$ should be DAGs or CPDAGs. While score-based causal discovery methods like GES \citep{chickering_optimal_2003}, NOTEARS \citep{zheng_dags_2018}, GLOBE \citep{mian_discovering_2021}, BCCD \citep{claassen_bayesian_2012}, parametric methods like LinGaM \citep{shimizu_linear_2006} or logic-based methods like \citep{hyttinen2014constraint} guarantee that their output is either a  proper causal graph or a MEC, this is no longer true for all PC-based algorithms. Due to assumption violations or contradictory independence test results during their execution, the order independent version of PC \citep{colombo_order-independent_2014} might run into logical conflicts between different orientation rules. or it might label unshielded triples $X-Y-Z$ in the graph as ambiguous. The conservative PC algorithm of \citep{ramsey_adjacency-faithfulness_2006} labels an unshielded triples $X-Y-Z$ ambiguous if the middle node $Y$ is in some but not all separating sets that the method found for $X$ and $Z$. The majority decision rule \citep{colombo_order-independent_2014} labels a triple ambiguous if $Y$ belongs to exactly half of all separating sets found for $X$ and $Z$. If conflicts or ambiguities occur, the output graph is \emph{invalid} in the sense that it does no longer imply separation/connection statements for arbitrary triples $(X,Y,\cS)$. Hence computing separation-based or adjustment-based metrics is no longer straightforward.

\paragraph{Dealing with Conflicting Orientations}

Since conflicts between different orientation rules are serious errors that flag assumption violations of some form, the most conservative way to treat them in a performance evaluation of a PC-like method is to record the proportion of their occurrence and then disregard graphs with conflicts for further steps. A well-performing method should lead to (a) a low proportion of graphs with conflicts and (b) good results w.r.t. the chosen evaluation metric on its valid output graphs.

\paragraph{Dealing with Ambiguities} 

\begin{algorithm}[t!] 
\begin{algorithmic} 
 \REQUIRE $\cG, \ \cH, \ \cU_{c}, \ \cU_{a}$ as defined above. A distance metric $d(\cdot,\cdot)$.
 \STATE
 \STATE Initialize empty list $\mathcal{L}$.
 \FORALL{$\cB \subset \cU_a$}
    \STATE $\cH_{\cB} \gets \cH$;
    \STATE $\mathrm{COL} \gets \cU_{c} \cup \cB$;
    \STATE Collider phase: \textbf{try}: orient all $(X,Y,Z) \in \mathrm{COL}$ as colliders on $\cH_{\cB}$;
    \IF{No conflicting orientations and no cycle in $\cH_{\cB}$}
    \STATE Orientation phase: \textbf{try}: Apply Meek's orientation rules to $\cH_{\cB}$;
    \IF{No conflicting orientations and no cycle in $\cH_{\cB}$}
    \STATE append $\cH_{\cB}$ to $\mathcal{L}$;
    \ENDIF
    \ENDIF
 \ENDFOR
 \IF{$\mathcal{L} = \emptyset$}
 \STATE max-dist, mean-dist, min-dist $\gets 1$;
 \ELSE
    \STATE max-dist $\gets \max_{\cB \in \mathcal{L}} d(\cG,\cH_{\cB})$;
    \STATE mean-dist $\gets \mathrm{mean}_{\cB \in \mathcal{L}} d(\cG,\cH_{\cB})$;
    \STATE min-dist $\gets \min_{\cB \in \mathcal{L}} d(\cG,\cH_{\cB})$;
 \ENDIF
\STATE \textbf{return} max-dist, mean-dist, min-dist.

\end{algorithmic}
\caption{Pseudocode to compute distance measures for graphical outputs with ambiguities.} \label{alg.ambiguities}
\end{algorithm}

In contrast to conflicting orientations, the appearance of ambiguities is a sign that a method is conservative about the inferences it is willing to make. Discarding graphs with ambiguities entirely, therefore seems like an overly harsh punishment for a method not being willing to make strong claims. We will now describe a procedure for computing best-case and worst-case distances to a ground truth for graphs with ambiguities by considering all possible ways in which ambiguities may be interpreted. The proposed strategy is similar to how the structural intervention distance is computed for a CPDAG by iterating through all DAGs that are represented by a given CPDAG, see \citep{peters_structural_2014}. The procedure will start with the following input data:

\begin{itemize}
    \item A (ground truth) DAG or CPDAG $\cG$;
    \item An undirected graph $\cH$ that will serve as the skeleton of a DAG or tsDAG;
    \item A partition of the set of unshielded triples $\cU = \cU_{c} \cup \cU_{nc} \cup \cU_{a} $ of $\cH$ into colliders ($ \cU_{c}$), non-colliders  ($ \cU_{nc}$) and ambiguous triples ($ \cU_{a}$).
\end{itemize}

This information is part of the output of the conservative PC algorithm \citep{ramsey_adjacency-faithfulness_2006}, the stable PC algorithm with the majority rule \citep{colombo_order-independent_2014} or the Tigramite implementation of PCMCI+ with the conservative or majority contemporary collider rule \citep{runge_discovering_2020}. 

From this input, we now generate a list of CPDAGs $(\cH_{\cB})$ where $\cB \subset \cU_{a}$. For each CPDAG $\cH_{\cB}$ in this list, we then compute the desired distance $d(\cG,\cH_{\cB})$ and, finally, we arrive at the best-case, average and worst-case estimate
\begin{align*}
    \min_{\cB \subset \cU_a} d_x(\cG,\cH_{\cB}), \qquad \mathrm{mean}(d_x(\cG,\cH_{\cB})) \qquad \max_{\cB \subset \cU_a} d_x(\cG,\cH_{\cB}).
\end{align*}

The details are outlined as pseudocode in Algorithm \ref{alg.ambiguities}.

\section{ADDITIONAL EXPERIMENTS} \label{app.additional_experiments}

In this section, we provide additional empirical experiments on separation-based distance measures.

\subsection{Correlations of symmetrized distance measures} \label{app.correlations}

In Figure \ref{fig.correlations}, we plot the correlation coefficients of the parent-SD with other distance measures for Erdös-Renyi DAGs $G(N,p)$ with $N= 25$ nodes for different values of $p$ ($500$ runs per parameter). We decided to use symmetrized versions of the included distance measures as this is a bit more similar to the SHD for a fair comparison. Parent-SD and its MB-enhanced variant are  strongly correlated for all values of $p$, suggesting that it might generically be advantageous to employ the latter due to its faster runtime. All other metrics are strongly correlated for very sparse graphs, but this correlation drops off rapidly and then increases again. Interestingly, for $0.25\leq p \leq 0.85$, parent-SD and SHD even become negatively correlated. These correlations do not differ substantially between DAGs and CPDAGs. We include a similar plot for the correlation of the ZL-SD and the SHD on mixed graphs further below in Appendix \ref{app.corr_ZLSD_SHD}. In addition, we compute the correlation of separation distances with the full separation metric that computes all separation statements on $N=10$ nodes in Appendix \ref{app.corr_with_sc_metric}. While this correlation is close to $1$ for sparse DAGs, it decreases when the DAGs become more dense.

 \begin{figure*}[h!]
    \centering
    \begin{subfigure}[t]{0.4\textwidth}
        \centering
        \includegraphics[height=1.6in]{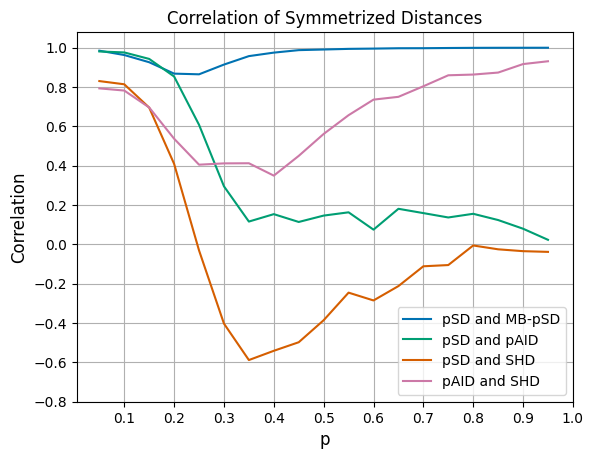}
    \end{subfigure}
    \begin{subfigure}[t]{0.4\textwidth}
        \centering
        \includegraphics[height=1.6in]{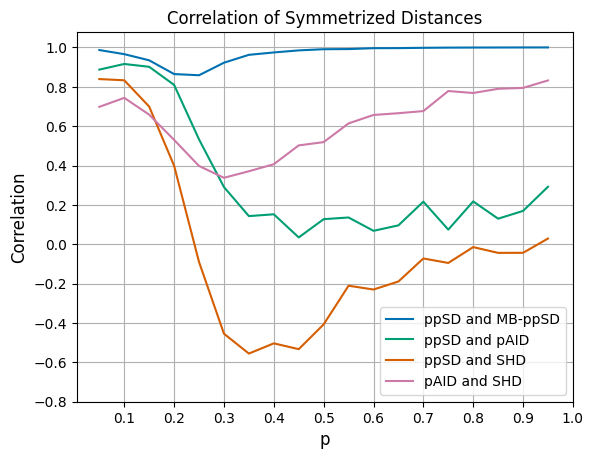}
    \end{subfigure}

    \caption{Correlation coefficients of symmetrized distance measures applied to Erdös-Renyi graphs $\cG,\cH \sim G(N,p)$ (left) and their associated CPDAGs (right) for increasing values of $p$. We ran $500$ experiments per parameter.} \label{fig.correlations}
\end{figure*}

\subsection{Correlation of ZL-SD and SHD on mixed graphs} \label{app.corr_ZLSD_SHD}

First, we repeat the experiments on the correlation of distance measures on mixed graphs. In this context, the available comparison metrics are the ZL-SD and the SHD, so these are the only ones we consider. We generate $500$ random mixed graphs by the following scheme. We first generate a causal order with a random permutation. Then for every node pair $(X,Y)$ for which $X<Y$ in the causal order, we generate an edge between them with probability $p$. If an edge is drawn, we orient the edge as $X\to Y$ with probability $q=0.2,0.7,0.9$ and as $X \leftrightarrow Y$ with probability $1-q$. The resulting graphs are acyclic but not necessarily MAGs as they might not be ancestral or might have almost cycles. The lack of ancestrality is unproblematic for the computation of the ZL-SD as the separator search will simply return \emph{None} for two non-adjacent nodes that cannot be separated. We did not apply any checks for almost-cycles as this would have introduced a significant computational overhead and would not have changed the general form of these plots.


\begin{figure*}[ht!]
    \centering
    \begin{subfigure}[t]{0.3\textwidth}
        \centering
        \includegraphics[height=1.6in]{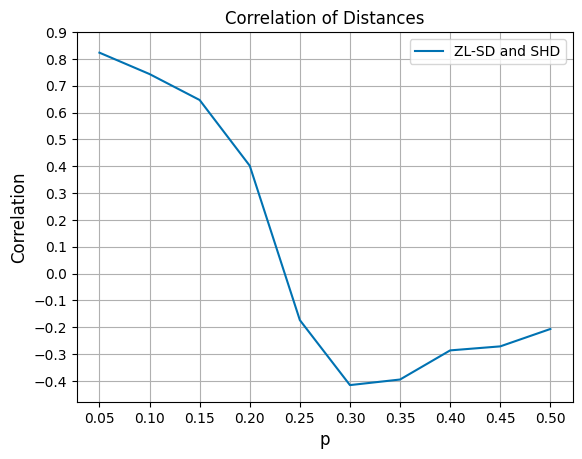}
    \end{subfigure}
    \begin{subfigure}[t]{0.3\textwidth}
        \centering
        \includegraphics[height=1.6in]{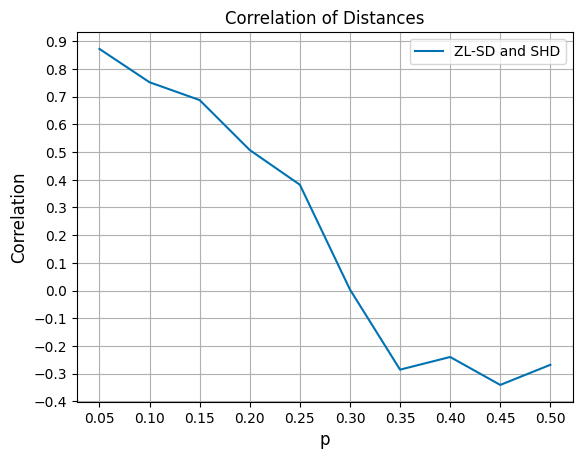}
    \end{subfigure}
    \begin{subfigure}[t]{0.3\textwidth}
        \centering
        \includegraphics[height=1.6in]{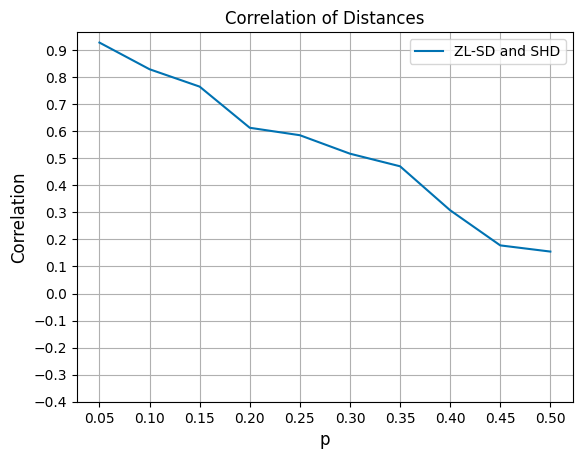}
    \end{subfigure}
    \caption{Correlation of ZL-SD and SHD across $500$ random mixed graphs with probability $q=0.2$ (left), $q=0.7$ (middle) and $q=0.9$ (right) that an existing edge is directed.}
    \label{fig.correlations_mixed}
\end{figure*}

\subsection{Correlation of SDs and the s-metric} \label{app.corr_with_sc_metric}
In this section, we present a plot on the correlation of separation distances with the full s-metric across $100$ pairs of Erdös-Renyi DAGs $\cG,\cH \sim G(N,p)$ on $10$ nodes for different values of $p$. More precisely, we compare the values of $d^{\mathfrak{S}}(\cG,\cH)$ to the values of $d^s(\cH,\cG)$ as these metrics are the most similar conceptually: in both cases separations found in $\cH$ are verified in $\cG$. We see that all SDs are strongly correlated with $d^s(\cH,\cG)$ for sparse graphs but the correlation drops for more dense ones. This illustrates that the choice of only specific separators according to a sep-strategies approximates a full comparison well on sparse graphs but the tradeoff of this selection becomes more apparent, the denser the graphs become.

\begin{figure}[t!]
    \centering
    \includegraphics[height=1.6in]{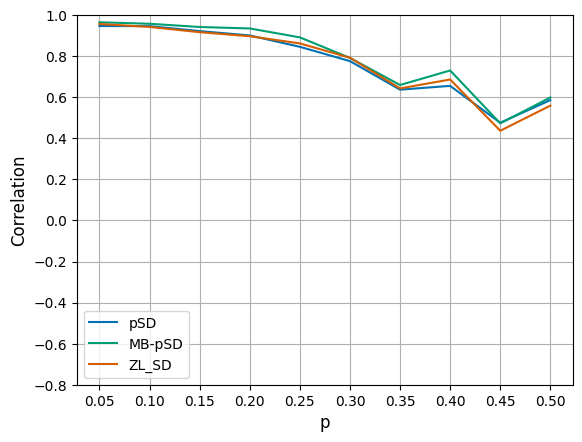}
    \caption{Correlation of SDs with the full s-metric  $d^s(\cH,\cG)$.}
   \label{fig.correlations_with_full}
\end{figure}

\subsection{Monte-Carlo approximation of the full s/c-metric}
Due to the exponential growth of the number of separation statements in the number of nodes $N$, computing the full s/c-distance as a weighted average across all separation statements is very time-consuming and only feasible for a small number of nodes. A possible way to alleviate this issue is to specify a fixed number $L$ of separation statements to be tested per order and to then draw these $L$ statements randomly. In this case the normalization constant in the order terms of the s/c-metric has to be set to $L^{-1}$. Figure \ref{fig.approximation_of_full} below shows that at least for small graphs, for which the full s/c-distance can still be checked, this approximation yields very good results. We note however that there is a trade-off between the quality of the approximation and the computational speed-up as high values of $L$ will give better approximation while low values will reduce the computational effort. Nevertheless, even with low values of $L$ the number of separation statements to be checked,  and hence the computational cost, is still significantly higher than for the deterministic sep-strategies. 

\begin{figure}[t!]
    \centering
    \includegraphics[height=1.6in]{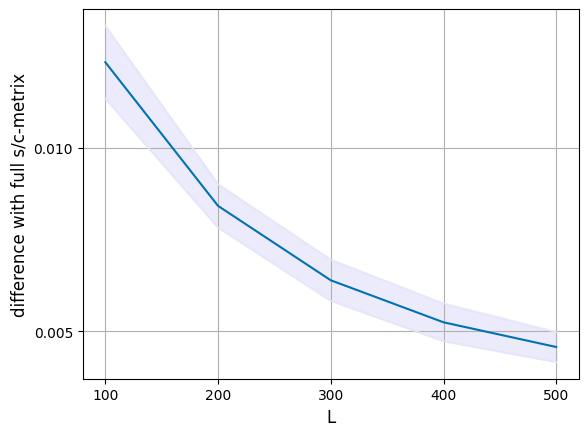}
    \caption{Difference to the full s/c-metric when only $L=100,200,300, 400,500$ separation statements are drawn randomly per order on Erdös-Renyi DAGs with $N=10$ nodes and edge probability $p=0.4$.}
   \label{fig.approximation_of_full}
\end{figure}

\paragraph{Reproducibility}
The Python code used for the experiments in this work is available in the repository \url{https://github.com/JonasChoice/SDs_experiments}.

\end{document}